\documentclass[11pt]{article}

\usepackage{fullpage}
\usepackage{amsmath}
\usepackage{amsfonts}
\usepackage{amssymb}
\usepackage{amsthm}
\usepackage{mathrsfs}
\usepackage{xspace}
\usepackage{booktabs}
\usepackage{array}
\usepackage{longtable}
\usepackage[dvipsnames]{xcolor}
\usepackage{hyperref}
\usepackage[nottoc]{tocbibind}

\numberwithin{equation}{section}
\allowdisplaybreaks

\newtheorem{theorem}{Theorem}
\newtheorem{lemma}{Lemma}
\newtheorem{proposition}{Proposition}

\newtheorem{claim}{Claim}
\theoremstyle{definition}
\newtheorem{definition}{Definition}
\newtheorem{example}{Example}
\theoremstyle{remark}

\newcommand{\SC}{\operatorname{SC}}
\newcommand{\OPT}{\operatorname{OPT}}

\newcommand{\PM}{\mathsf{P}}
\newcommand{\PD}{\mathsf{D}}
\newcommand{\Mix}{\mathsf{M}}
\newcommand{\HD}{\mathsf{HD}}

\newcommand{\E}{\mathbb{E}}

\newcommand{\pos}[1]{\left[#1\right]_{+}}
\newcommand{\eps}{\varepsilon}

\title{Randomized Strategyproof Facility Location:\\ Two Facilities and Beyond}

\author{Hau Chan$^{1}$ \qquad Jianan Lin$^{2}$ \qquad Chenhao Wang$^{3,4}$\\[0.6em]
\small $^{1}$University of Nebraska--Lincoln\\
\small $^{2}$Rensselaer Polytechnic Institute\\
\small $^{3}$Beijing Normal University--Zhuhai\\
\small $^{4}$Beijing Normal--Hong Kong Baptist University}

\date{}

\begin{document}
\maketitle

\begin{abstract}
We design and analyze randomized strategyproof mechanisms for multi-facility location under the utilitarian social-cost objective, the sum of the agents' distances to their nearest facilities.
For two facilities, the Pairwise-Distance mechanism locates facilities at a pair of reported locations sampled with probability proportional to their distance.
It is strategyproof on Ptolemaic spaces, including Euclidean and Hilbert spaces as special cases, and has an approximation ratio of \(4\).
The resulting Hybrid-Distance mechanism is a fixed-probability mixture: it selects the classical Proportional mechanism [Lu et al., EC'10] with probability \(\lambda^*=\frac{5+4\sqrt3}{23}\) and Pairwise-Distance with probability $1-\lambda^*$.
It is strategyproof on Ptolemaic spaces and has a tight approximation ratio of \(\frac{74+4\sqrt3}{23}\approx3.5186\), breaking the long-standing factor-\(4\) benchmark of [Lu et al., EC'10].

We complement the two-facility results by studying more facilities.
First, for \(n\) agents and \(k=n-1\) facilities, we introduce the Inverse-Square mechanism, which omits one report with probability proportional to the inverse square of its nearest-neighbor distance and locates facilities at all remaining reports.
It is strategyproof on any metric space and has an approximation ratio of \(\Theta(\sqrt{n})\), improving the previous best-known ratio of \(\frac{n}{2}\) [Escoffier et al., ADT'11].
Second, for $k$ facilities on the line, we introduce the Gap-Product mechanism, which locates facilities at \(k\) reports and weights each set by the product of the gaps between consecutive selected reports.
When \(k=3\), it is strategyproof and has a \(6\)-approximation, replacing the previous \(n\)-dependent guarantee [Fotakis and Tzamos, EC'13] by a constant. 
\end{abstract}

\section{Introduction}\label{sec:introduction}

Facility location is a canonical problem in approximate mechanism design without money and a fundamental optimization problem in operations research and theoretical computer science \cite{procaccia2013approximate,owen1998strategic,drezner2004facility}.
A social planner locates one or more facilities in a metric space using locations reported by agents, and each agent incurs the distance from her true location to the nearest facility.
Under the utilitarian social-cost objective, the planner seeks to minimize the sum of these distances.
Because an agent may benefit by misreporting her location, the central challenge is to design a strategyproof mechanism that elicits truthful location reports while retaining a good approximation to the optimal social cost.
Such mechanisms can inform the placement of public resources such as schools, parks, and libraries \cite{drezner2004facility}.
As preference-aggregation rules, they can also select outcomes or representatives that reflect a group's political preferences \cite{black1948rationale,moulin1980strategy}.

The classical single-facility problem on the line provides an important contrast: locating the facility at a median report is strategyproof and optimal for the utilitarian social cost \cite{moulin1980strategy,procaccia2013approximate}.
For the egalitarian objective, which minimizes the maximum individual cost, no deterministic strategyproof mechanism can beat factor \(2\), whereas randomization yields the optimal factor \(\frac{3}{2}\) \cite{procaccia2013approximate}.
The single-facility literature has also studied alternative objectives, including \(L_p\)-norm, least-squares, and utility objectives \cite{feigenbaum2017approximately,feldman2013strategyproof,walsh24utility}, and domains such as trees \cite{alon2010strategyproof}, circles \cite{DBLP:conf/sagt/Meir19}, and higher-dimensional spaces \cite{GoelH23,gravin2025approximation,barak2026facility,chan2026strategyproof}.

\paragraph{Multi-Facility Location under the Utilitarian Social-Cost Objective.}
Locating $k\ge 2$ facilities changes both the optimization problem and its incentive constraints: an agent's cost becomes the minimum of several distances, while the mechanism must coordinate several reported locations.
For \(n\) agents and \(k=2\) facilities on the line, the deterministic Two-Extremes mechanism by Procaccia and Tennenholtz \cite{procaccia2013approximate} locates facilities at the leftmost and rightmost reported positions and has an approximation ratio of \(n-2\).
For \(n\ge5\), no deterministic strategyproof mechanism can achieve a ratio better than \(n-2\) \cite{fotakis2014power}.
The deterministic picture becomes even more restrictive beyond two facilities: for every \(k\ge3\), no deterministic anonymous strategyproof mechanism has a bounded approximation ratio even on the line \cite{fotakis2014power}.

Randomization provides substantially better guarantees across several multi-facility regimes.
For two facilities, Lu et al.~\cite{lu2009tighter} established early randomized bounds, and Lu et al.~\cite{lu10mechanism} subsequently introduced the Proportional mechanism.
This mechanism chooses an anchor uniformly from the reported locations and then chooses a second report with probability proportional to its distance from the anchor.
It is strategyproof on every metric space and has an approximation ratio of \(4\).
For arbitrary \(2\le k<n\) on the line, the Equal Cost mechanism by Fotakis and Tzamos~\cite{fotakis2013strategyproof} is group-strategyproof and has an approximation ratio of at most \(n\) for the utilitarian social cost.
For the special setting \(k=n-1\), the Pick the Loser mechanism randomly leaves one report without a facility; on the line, it is strongly group-strategyproof and has an approximation ratio of \(2\) \cite{fotakis2013strategyproof}.
For \(k=n-1\) on arbitrary metric spaces, Escoffier et al.~\cite{escoffier2011many} proposed the Inversely Proportional mechanism, which omits report \(i\) with probability inversely proportional to its nearest-neighbor distance.
It is strategyproof and has a tight approximation ratio of \(\frac{n}{2}\).

These results suggest that the effectiveness of randomization depends jointly on the number of facilities and the geometry of the location space, leading to our central question:

\begin{quote}
\emph{How do the number of facilities and the geometry of the location space affect the approximation guarantees achievable by randomized strategyproof mechanisms for multi-facility location?}
\end{quote}

We investigate this question through three complementary settings.
For two facilities, we focus on Ptolemaic spaces, a class containing Euclidean and Hilbert spaces, metric trees, and \(\mathrm{CAT}(0)\) spaces \cite{foertsch2007ptolemy}; their defining four-point inequality provides geometric structure for analyzing pairwise-distance randomization.
For \(n\) agents and \(k=n-1\) facilities, we work on arbitrary metric spaces, where the mechanism needs only to randomize over the single omitted report.
Finally, on the line, we study gap-based randomization for three facilities together with its approximation behavior for general \(k\).

\subsection{Our Contributions}\label{subsec:contributions}

Our contributions are as follows.

\begin{itemize}
\item \textbf{Two facilities in Ptolemaic spaces.} We introduce the \emph{Pairwise-Distance} mechanism, which locates facilities at an unordered pair of reports sampled with probability proportional to the distance between them.
It is strategyproof on every Ptolemaic metric space and has an approximation ratio of \(4\), with the approximation guarantee holding even on arbitrary metric spaces.
We then consider fixed-probability mixtures of Pairwise-Distance and the Proportional mechanism of Lu et al.~\cite{lu10mechanism}.
The resulting \emph{Hybrid-Distance} mechanism selects Proportional with probability \(\lambda^*=\frac{5+4\sqrt3}{23}\), is strategyproof on Ptolemaic spaces, and has an approximation ratio of \(3+\lambda^*=\frac{74+4\sqrt3}{23}\approx 3.5186\). 
Two families of line instances establish the tightness of the analysis.
Furthermore, the choice of \(\lambda^*\) optimizes the approximation ratio among all profile-independent mixtures of these two component mechanisms.

\item \textbf{\(n-1\) facilities in arbitrary metric spaces.} For \(n\) agents and \(k=n-1\) facilities, we introduce the \emph{Inverse-Square} mechanism, which omits one report with probability proportional to the inverse square of its distance to the nearest other reported location and locates facilities at all remaining reports.
It is strategyproof on every metric space and has a tight approximation ratio of \(\frac{1+\sqrt{n/2}}{2}=\Theta(\sqrt{n})\).
This improves the \(\frac{n}{2}\)-approximation of the Inversely Proportional mechanism of Escoffier et al.~\cite{escoffier2011many} in the same arbitrary-metric setting.
While the Pick the Loser mechanism of Fotakis and Tzamos \cite{fotakis2013strategyproof} achieves the stronger ratio \(2\) on the line, the contribution of Inverse-Square is instead its geometry-independent improvement over the corresponding general-metric baseline.

\item \textbf{Three facilities on the line.} We introduce the \emph{Gap-Product} mechanism on the line, which locates facilities at \(k\) selected reports and assigns each selection a weight equal to the product of the gaps between consecutive selected locations in left-to-right order. It reduces to Pairwise-Distance when $k=2$.
For \(k=3\), we prove that it is strategyproof with a tight approximation ratio of \(6\), replacing the previous \(n\)-approximation \cite{fotakis2013strategyproof} by a constant. 
More generally, Gap-Product has a tight approximation ratio of \(2k\) for every \(k\ge2\), although strategyproofness fails for every \(k\ge4\).
The approximation proof uses an equivalent random partition into consecutive blocks, and the incentive proof groups the deviation expansion by overlap type.
\end{itemize}

\paragraph{Concurrent work.} We note that there are {three} concurrent and independent works.
First, for two facilities, Ma and Peng \cite{ma2026breaking} independently proposed the Pairwise-Distance mechanism (Global Pair in their terminology) and proved its strategyproofness in Ptolemaic spaces.
They showed that the mixture that selects Pairwise-Distance with probability $\frac{1}{3}$ and Proportional with probability $\frac{2}{3}$ achieves an approximation ratio of $3.6667$.
They further proved a lower bound of $3.5186$ on the approximation ratio of every fixed mixture of these two mechanisms and a lower bound of $1.2071$ for any randomized strategyproof mechanism.
In Section~\ref{sec:two-facility}, we establish the exact approximation ratio of $3.5186$ for the mixture and thereby close the gap, using a different accounting and tightness analysis.

Second, for three facilities on the line, Aziz, Mackenzie, and Suzuki \cite{aziz2026anchoring} introduced the \emph{Random-Anchor Volume} mechanism, which opens an anchor facility at a uniformly random report and selects two more reports with probability proportional to the product of their consecutive gaps from the anchor.
They showed that it is strategyproof and achieves an $8$-approximation.
They also noted that the Gap-Product mechanism achieves a \(2k\)-approximation, while leaving its strategyproofness for \(k=3\) as an open question.
In Section~\ref{sec:k-facility}, we prove its strategyproofness and the  $6$-approximation for $k=3$.

Third, Jia \cite{jia2026product} independently introduced Gap-Product and proved the same \(2k\)-approximation for every \(k\ge2\) and the strategyproofness for \(k=3\). When $k=2$,
Jia also found that its random mixture with Proportional can achieve the \(3.5186\)-approximation, the same with ours. 
While the focus of \cite{jia2026product} is only on the line, our two-facility strategyproofness results extend to Ptolemaic spaces, and our two-facility approximation bounds hold on arbitrary metric spaces

Moreover, compared to these work, our $\Theta(\sqrt{n})$-approximation result for $k=n-1$ in all metric spaces is totally new.

\subsection{Other Related Work}\label{subsec:related}


Beyond the standard social-cost benchmarks reviewed above, the multi-facility literature has considered alternative connection costs, service rules, preference models, and manipulation models.
Fotakis and Tzamos studied concave connection costs and winner-imposing mechanisms for multiple facilities on the line \cite{fotakis2013strategyproof,fotakis2013winner}.
Other work allows agents to have heterogeneous or fractional preferences over facilities, with much of this literature focusing on two facilities or limited service resources \cite{PaoloCarmine2016,fong2018facility}.
Related two-facility variants include opposite facilities \cite{chen2018mechanism,chen2021tight} and false-name-proof mechanisms \cite{sonoda2016false}.
These models share the presence of multiple facilities with our setting but modify agents' preferences, their service rule, or the permitted form of manipulation.

Other variants consider attractive, obnoxious, or hybrid facilities \cite{cheng2013strategy,feigenbaum2015strategyproof,chan2025obnoxious}, or imposes capacity constraints on how agents can be assigned to facilities \cite{aziz2020facility,aziz2020capacity}.
Fairness-oriented variants replace or supplement aggregate efficiency with individual, group, or equity guarantees \cite{alex2024,Zhou22Group-Fair,walsh2025equitable}.
Learning-augmented mechanisms additionally receive a prediction or advice signal and balance consistency with robustness and strategyproofness \cite{agrawal2022learning,Xu2022,balkanski2024randomized,DBLP:conf/nips/Barak0T24,chan2025prediction,nips2025envy}.
The survey of Chan et al.\ \cite{chan2021mechanismsurvey} provides a broader account of strategyproof facility location.

\section{Preliminaries}\label{sec:preliminaries}

Let \((\Omega,d)\) be a metric space and fix a facility budget \(k\ge1\).
For each population size \(n\ge1\), let \(N=\{1,\ldots,n\}\) be the set of agents.
Agent \(i\) is located at \(x_i\in\Omega\), and
\(\mathbf{x}=(x_1,\ldots,x_n)\in\Omega^n\) denotes the location profile.
Let \(\mathcal F_k=\{F\subseteq\Omega:1\le |F|\le k\}\),
and let \(\Delta(\mathcal F_k)\) denote the set of probability distributions over \(\mathcal F_k\).

A randomized \(k\)-facility mechanism \(f\) specifies, for every population size \(n\ge1\), a function \(f_n:\Omega^n\longrightarrow\Delta(\mathcal F_k)\).
We suppress the subscript \(n\) when the population size is clear.
On input \(\mathbf{x}\), the mechanism draws a facility set \(F\sim f(\mathbf{x})\).

For \(F\in\mathcal F_k\), the cost of agent \(i\) and the social cost at profile \(\mathbf{x}\) are
\[
 c_i(F,\mathbf{x})=\min_{y\in F}d(x_i,y),
 \qquad
 \SC(F,\mathbf{x})=\sum_{i\in N}c_i(F,\mathbf{x}).
\]
For a randomized mechanism \(f\), define
\begin{align*}
 c_i(f(\mathbf{x}),\mathbf{x})
 &=\E_{F\sim f(\mathbf{x})}[c_i(F,\mathbf{x})],\\
 \SC(f(\mathbf{x}),\mathbf{x})
 &=\E_{F\sim f(\mathbf{x})}[\SC(F,\mathbf{x})]
 =\sum_{i\in N}c_i(f(\mathbf{x}),\mathbf{x}).
\end{align*}

The optimal social cost with at most \(k\) facilities is
\[
 \OPT_k(\mathbf{x})=\inf_{F\in\mathcal F_k}\SC(F,\mathbf{x}).
\]
When the profile is fixed, we write \(c_i(f)\) and \(\SC(f)\) for
\(c_i(f(\mathbf{x}),\mathbf{x})\) and \(\SC(f(\mathbf{x}),\mathbf{x})\), respectively.

\begin{definition}[Strategyproofness]
A randomized mechanism \(f\) is \emph{strategyproof} if, for every \(n\ge1\), every profile \(\mathbf{x}\in\Omega^n\), every agent \(i\in N\), and every alternative report \(x'_i\in\Omega\),
\[
 c_i(f(\mathbf{x}),\mathbf{x})
 \le
 c_i\bigl(f(x'_i,\mathbf{x}_{-i}),\mathbf{x}\bigr).
\]
Thus, no agent can reduce her expected cost by misreporting her location while the other reports remain fixed.
\end{definition}

A randomized \(k\)-facility mechanism \(f\) has approximation ratio at most \(\rho\) if
\[
 \SC(f(\mathbf{x}),\mathbf{x})\le \rho\,\OPT_k(\mathbf{x})
\]
for every population size \(n\ge1\) and every profile \(\mathbf{x}\in\Omega^n\) with \(\OPT_k(\mathbf{x})>0\).
When the mechanism is considered on more than one metric space, the inequality is required on every such space.
The \emph{approximation ratio} \(\operatorname{apx}(f)\) is the infimum over all \(\rho\ge1\) satisfying this uniform bound.

We are particularly interested in Ptolemaic metric spaces. Every Euclidean or Hilbert space, and every metric tree is Ptolemaic \cite{foertsch2007ptolemy}.
Hence all results proved below for Ptolemaic spaces apply to these familiar classes.

\begin{definition}[Ptolemaic space]\label{def:ptolemaic}
A metric space $(\Omega,d)$ is \emph{Ptolemaic} if every four points \(u,v,w,z\in\Omega\) satisfy
\[
 d(u,w)d(v,z)
 \le
 d(u,v)d(w,z)+d(u,z)d(v,w).
\]
\end{definition}

\section{Two Facilities}
\label{sec:two-facility}

Throughout this section, \(k=2\), and we write \(\OPT=\OPT_2\).
For a profile \(\mathbf{x}\), let
\[
 d_{ij}=d(x_i,x_j),\qquad
 S_i=\sum_{j\in N}d_{ij},\qquad
 Z=\sum_{i<j}d_{ij}.
\]
Thus, \(d_{ij}\) is the distance between \(x_i\) and \(x_j\), \(S_i\) is the sum of distances from \(x_i\) to the profile, and \(Z\) is the total distance over all unordered pairs.
In particular, \(\sum_iS_i=2Z\).

We study two mechanisms that use the same distance-proportional second draw but different distributions over the first agent.

\begin{definition}[Proportional mechanism \(\PM\) \cite{lu10mechanism}]\label{def:proportional}
If \(Z=0\), open a facility at the common reported location.
Otherwise, choose an anchor \(i\) uniformly from \(N\), choose a second agent \(j\ne i\) with probability \(d_{ij}/S_i\), and open facilities at \(x_i\) and \(x_j\).
\end{definition}

\begin{definition}[Pairwise-Distance mechanism \(\PD\)]
\label{def:pairwise-distance}
If \(Z=0\), open a facility at the common reported location.
Otherwise, choose an unordered pair \(\{i,j\}\) with probability \(d_{ij}/Z\), and open facilities at \(x_i\) and \(x_j\).
\end{definition}

Pairwise-Distance can equivalently be implemented by choosing the anchor \(i\) with probability \(S_i/(2Z)\) and then choosing \(j\) with probability \(d_{ij}/S_i\).
Thus, the two mechanisms differ only in their anchor distributions: Proportional uses the uniform distribution, whereas Pairwise-Distance favors reports that are far from the rest of the profile.

Lu et al.~\cite{lu10mechanism} proved that Proportional is strategyproof and a \(4\)-approximation on every metric space.
As we will show later, Pairwise-Distance is strategyproof on every Ptolemaic space, and achieves a \(4\)-approximation.
Thus, an immediate question is whether a mixture of the two mechanisms can improve the approximation guarantee.

\begin{definition}[Fixed mixture]\label{def:mixture}
For \(\lambda\in[0,1]\), the mechanism
\(
 \Mix_\lambda:=\lambda\PM+(1-\lambda)\PD
\)
runs Proportional with probability \(\lambda\) and Pairwise-Distance with probability \(1-\lambda\), independently of the profile.
\end{definition}

Define
\[
 \lambda^*=\frac{1}{4\sqrt3-5}=\frac{5+4\sqrt3}{23}.
\]
We call the specific weighted mixture
\[
 \HD:=\Mix_{\lambda^*}
 =\lambda^*\PM+(1-\lambda^*)\PD
\]
the \emph{Hybrid-Distance mechanism}.

\begin{theorem}\label{thm:incentive-properties}
The Hybrid-Distance mechanism \(\HD\) is strategyproof on every Ptolemaic metric space.
\end{theorem}

Our main result determines the exact approximation ratio of Hybrid-Distance and its optimality within the fixed-mixture family.

\begin{theorem}\label{thm:main}
For every population size and every profile on every metric space,
\[
 \SC(\HD)\le(3+\lambda^*)\OPT.
\]
The bound is asymptotically tight on the line, and hence the approximation ratio of \(\HD\) is exactly \(3+\lambda^*\).
Moreover, \(\lambda^*\) minimizes the worst-case approximation ratio over the family \(\{\Mix_\lambda:\lambda\in[0,1]\}\).
\end{theorem}

The remainder of the section has three parts.
Subsection~\ref{sec:pd-strategyproofness} proves Theorem~\ref{thm:incentive-properties}.
Subsection~\ref{sec:pd-excess} proves the upper bound for \(\HD\).
Finally, Subsection~\ref{sec:tightness} proves tightness and determines the optimal mixing weight, completing the proof of Theorem~\ref{thm:main}.

\subsection{Strategyproofness}
\label{sec:pd-strategyproofness}

We now prove Theorem~\ref{thm:incentive-properties}. Hybrid-Distance uses a report-independent mixing probability, so it is strategyproof if and only if the two components are both strategyproof.
The strategyproofness of Proportional is known by Lu et al.~\cite{lu10mechanism}, so only the Pairwise-Distance component requires an argument.

\begin{proof}[Proof of Theorem~\ref{thm:incentive-properties}]
We prove that Pairwise-Distance is strategyproof on every Ptolemaic metric space.
Fix a profile in which the agent under consideration is located at \(x\), let \(z\) be an alternative report, and denote the other agents' locations by \(y_1,\ldots,y_m\).
For simplicity, define
\[
 \begin{aligned}
 B&=\sum_{i<j}d(y_i,y_j),
 &A&=\sum_{i<j}d(y_i,y_j)\min\{d(x,y_i),d(x,y_j)\},\\
 D_x&=\sum_i d(x,y_i),
 &D_z&=\sum_i d(z,y_i),\\
 E&=\sum_i d(z,y_i)\min\{d(x,z),d(x,y_i)\}.
 \end{aligned}
\]
Here, \(B\) is the total sampling weight of pairs containing only other agents, and \(A\) is the contribution of these pairs to the expected-cost numerator of the agent at \(x\).
Moreover, \(D_x\) and \(D_z\) are the total sampling weights of pairs containing, respectively, the report \(x\) and the report \(z\), while \(E\) is the contribution of the latter pairs to the expected-cost numerator of the agent at \(x\).
In particular, \(A\) and \(B\) are independent of whether the agent reports \(x\) or \(z\).
If \(B+D_x=0\), truthful reporting gives the agent cost zero, so the desired inequality is immediate.
Assume henceforth that \(B+D_x>0\).
The agent's expected costs under truthful reporting and reporting \(z\) are, respectively,
\[
 \frac{A}{B+D_x}
 \qquad\text{and}\qquad
 \frac{A+E}{B+D_z}.
\]
After cross-multiplication, strategyproofness is equivalent to
\(
 E(B+D_x)\ge A(D_z-D_x).
\)
It is therefore sufficient to prove the stronger inequality
\begin{equation}\label{eq:sp-strong}
 BE\ge A(D_z-D_x).
\end{equation}
The advantage of \eqref{eq:sp-strong} is that its left-hand side admits a decomposition into local terms indexed by pairs and triples.

Set
\[
 \delta_i=d(z,y_i)-d(x,y_i),\qquad
 e_i=d(z,y_i)\min\{d(x,z),d(x,y_i)\},
\]
and \(m_{ij}=\min\{d(x,y_i),d(x,y_j)\}\).
Expanding the difference in \eqref{eq:sp-strong} gives
\begin{equation}\label{eq:sp-expansion}
 BE-A(D_z-D_x)
 =\sum_{i<j}\sum_k
   d(y_i,y_j)\bigl(e_k-m_{ij}\delta_k\bigr).
\end{equation}
We split this double sum into endpoint terms, where \(k\in\{i,j\}\), and triple terms, where \(i,j,k\) are distinct.
The endpoint terms need only the triangle inequality; Ptolemy is used solely to control each group of three terms belonging to one unordered triple.
If \(k\in\{i,j\}\), then \(0\le m_{ij}\le d(x,y_k)\) and
\begin{equation}\label{eq:endpoint-nonnegative}
 e_k-m_{ij}\delta_k\ge0.
\end{equation}
Indeed, this is immediate when \(\delta_k\le0\).
When \(\delta_k>0\), the triangle inequality gives \(\delta_k\le d(x,z)\).
If \(d(x,y_k)\le d(x,z)\), then
\(e_k=d(z,y_k)d(x,y_k)\ge m_{ij}\delta_k\) because \(d(z,y_k)\ge\delta_k\).
If \(d(x,z)<d(x,y_k)\), then
\(e_k=d(z,y_k)d(x,z)\ge m_{ij}\delta_k\) because
\(d(x,z)\ge\delta_k\) and \(d(z,y_k)>d(x,y_k)\ge m_{ij}\).

It remains to group the terms in \eqref{eq:sp-expansion} having three distinct indices.
The following lemma is the required local inequality.

\begin{lemma}\label{lem:sp-triple}
For \(i\in\{1,2,3\}\), let \(d_i=d(x,y_i)\), and relabel the three points so that \(d_1\le d_2\le d_3\).
Then 
\begin{equation}\label{eq:sp-triple}
 \begin{split}
 T:={}&d(y_1,y_2)(e_3-d_1\delta_3)
       +d(y_1,y_3)(e_2-d_1\delta_2)+d(y_2,y_3)(e_1-d_2\delta_1)\ge0.
 \end{split}
\end{equation}
\end{lemma}

\begin{proof}
We distinguish the position of \(d(x,z)\) relative to \(d_1\le d_2\le d_3\).

If \(d(x,z)\le d_1\), then \(\min\{d(x,z),d(x,y_i)\}=d(x,z)\) for all \(i\).
Consequently,
\[
 \begin{aligned}
 e_3-d_1\delta_3
 &=d_1d_3+(d(x,z)-d_1)d(z,y_3)\\
 &\ge d_1d_3+(d(x,z)-d_1)(d_3+d(x,z))\\
 &=d(x,z)(d_3+d(x,z)-d_1),\\
 e_2-d_1\delta_2
 &=d_1d_2+(d(x,z)-d_1)d(z,y_2)\\
 &\ge d_1d_2+(d(x,z)-d_1)(d_2+d(x,z))\\
 &=d(x,z)(d_2+d(x,z)-d_1),\\
 e_1-d_2\delta_1
 &=d_1d_2+(d(x,z)-d_2)d(z,y_1)\\
 &\ge d_1d_2+(d(x,z)-d_2)(d_1+d(x,z))\\
 &=d(x,z)(d_1+d(x,z)-d_2).
 \end{aligned}
\]
Each inequality follows from \(d(z,y_i)\le d(x,y_i)+d(x,z)\): the coefficient multiplying \(d(z,y_i)\) is nonpositive because \(d(x,z)\le d_1\le d_2\le d_3\).
Multiplying these three inequalities by the corresponding nonnegative pair distances and summing gives
\[
 \begin{split}
 T\ge d(x,z)\{&d(y_1,y_2)(d_3+d(x,z)-d_1)
                 +d(y_1,y_3)(d_2+d(x,z)-d_1)\\
               &+d(y_2,y_3)(d_1+d(x,z)-d_2)\}.
 \end{split}
\]
The expression is nonnegative if \(d_1+d(x,z)-d_2\ge0\).
Otherwise \(d(y_2,y_3)\le d(y_1,y_2)+d(y_1,y_3)\), and the expression in braces is at least
\[
 d(y_1,y_2)(d_3-d_2+2d(x,z))+2d(x,z)d(y_1,y_3)\ge0.
\]

If \(d_1\le d(x,z)\le d_2\),
then
\begin{align*}
 T={}&d(y_1,y_2)\{d_1d_3+(d(x,z)-d_1)d(z,y_3)\}+d(y_1,y_3)\{d_1d_2+(d(x,z)-d_1)d(z,y_2)\}\\
     &+d(y_2,y_3)\{d_1d_2-(d_2-d_1)d(z,y_1)\}.
\end{align*}
If the last brace is nonnegative, all three terms are nonnegative and hence \(T\ge0\).
We may therefore assume that the last brace is negative, and set
\[
 K=(d_2-d_1)d(z,y_1)-d_1d_2>0,\qquad
 K_0=d(x,z)(d_2-d_1)-d_1^2.
\]
Denote the first two positive braces by
\(U=d_1d_3+(d(x,z)-d_1)d(z,y_3)\) and
\(V=d_1d_2+(d(x,z)-d_1)d(z,y_2)\), so that
\[
 T=d(y_1,y_2)U+d(y_1,y_3)V-d(y_2,y_3)K.
\]
Since \(d(z,y_1)\le d_1+d(x,z)\), we have \(K\le K_0\).
The triangle inequality also gives
\(
 d(y_1,y_2)\ge d_2-d_1, d(y_1,y_3)\ge d_3-d_1
\)
and
\(
 d(y_2,y_3)\le M:=\min\{d_2+d_3,d(z,y_2)+d(z,y_3)\}.
\)
Therefore
\[
 T\ge (d_2-d_1)U+(d_3-d_1)V-MK_0.
\]
Consequently it suffices to verify
\begin{equation}\label{eq:sp-middle-claim}
 \begin{split}
 L:={}&(d_2-d_1)\{d_1d_3+(d(x,z)-d_1)d(z,y_3)\}\\
      &+(d_3-d_1)\{d_1d_2+(d(x,z)-d_1)d(z,y_2)\}
        \ge K_0M.
 \end{split}
\end{equation}
It remains to verify this scalar inequality.
For \(i\in\{2,3\}\), define the slack in the lower triangle bound by
\[
 s_i=d(z,y_i)-\bigl(d_i-d(x,z)\bigr).
\]
The triangle inequality gives \(0\le s_i\le2d(x,z)\), and
\[
 M=d_2+d_3-2d(x,z)
   +\min\{s_2+s_3,2d(x,z)\}.
\]
Hence \(L-K_0M\) is affine in $s_2$ and $s_3$ on each side of the line
\(s_2+s_3=2d(x,z)\).
On the side \(s_2+s_3\le2d(x,z)\), since $s_2,s_3\ge 0$, the minimum of \(L-K_0M\) is attained at one of
\[
 (s_2,s_3)\in\{(0,0),\quad(2d(x,z),0),\quad(0,2d(x,z))\}.
\]
At these three points, direct substitution gives, respectively,
\begin{align*}
 &2d_1\{d(x,z)(d_2-d_1)+d_1(d_2-d(x,z))\}+(d_3-d_2)\{d(x,z)(d_1+d_2-d(x,z))+d_1^2\},\\
 &2d_1^2d_2
   +(d_3-d_2)\{d(x,z)(d_2+d(x,z)-d_1)+d_1^2\},\\
 &2d_1^2d_2
   +(d_3-d_2)\{d(x,z)(d_1+d_2-d(x,z))+d_1^2\}.
\end{align*}
Every displayed quantity is nonnegative because
\(d_1\le d(x,z)\le d_2\le d_3\).

On the side \(s_2+s_3\ge2d(x,z)\), the value of \(M\) is constant.
The coefficients of \(s_2\) and \(s_3\) in \(L\) are, respectively,
\[
 (d_3-d_1)(d(x,z)-d_1)
 \quad\text{and}\quad
 (d_2-d_1)(d(x,z)-d_1).
\]
The first is at least the second, and both are nonnegative.
Thus the minimum on this side is attained at
\((s_2,s_3)=(0,2d(x,z))\), which is already the third point checked above.
This proves \eqref{eq:sp-middle-claim}.

Finally, if \(d(x,z)\ge d_2\),
Ptolemy's inequality on the four points \(z,y_1,y_2,y_3\) gives
\[
 \Delta:=d(y_1,y_2)d(z,y_3)+d(y_1,y_3)d(z,y_2)
          -d(y_2,y_3)d(z,y_1)\ge0.
\]
If \(d_2\le d(x,z)\le d_3\), direct substitution of the three minima \((d_1,d_2,d(x,z))\) yields
\[
 \begin{split}
 T={}&(d_2-d_1)\Delta
       +d(y_1,y_2)\{d_1d_3+(d(x,z)-d_2)d(z,y_3)\}\\
     &+d_1d_2\{d(y_1,y_3)+d(y_2,y_3)\}\ge0.
 \end{split}
\]
If \(d(x,z)\ge d_3\), substituting the three minima \((d_1,d_2,d_3)\) instead yields
\[
 \begin{split}
 T={}&(d_2-d_1)\Delta
       +d(y_1,y_2)\{d_1d_3+(d_3-d_2)d(z,y_3)\}\\
     &+d_1d_2\{d(y_1,y_3)+d(y_2,y_3)\}\ge0.
 \end{split}
\]
These cases exhaust all possibilities.
\end{proof}

Fix an unordered triple of indices and relabel it as \(\{1,2,3\}\) so that
\(d_1=d(x,y_1)\le d_2= d(x,y_2)\le d_3= d(x,y_3)\).
The three distinct-index terms in \eqref{eq:sp-expansion} associated with this triple are
\[
 \begin{split}
 &d(y_1,y_2)(e_3-m_{12}\delta_3)
 +d(y_1,y_3)(e_2-m_{13}\delta_2)+d(y_2,y_3)(e_1-m_{23}\delta_1).
 \end{split}
\]
By the definition of \(m_{ij}\) and the chosen ordering,
\(
 m_{12}=m_{13}=d_1, m_{23}=d_2.
\)
Thus, the preceding sum is exactly the quantity \(T\) in \eqref{eq:sp-triple}, and is nonnegative by Lemma~\ref{lem:sp-triple}.
Since this argument applies to every unordered triple, all distinct-index terms in \eqref{eq:sp-expansion} are nonnegative; all endpoint terms are nonnegative by \eqref{eq:endpoint-nonnegative}.
Hence \eqref{eq:sp-strong} holds, so a misreport cannot reduce the agent's expected cost.
This proves Pairwise-Distance strategyproofness on every Ptolemaic metric space.
Hybrid-Distance is a report-independent randomization between Proportional and Pairwise-Distance, so it is strategyproof as well.
\end{proof}

\subsection{Approximation ratio}
\label{sec:pd-excess}

We now prove the approximation part of Theorem~\ref{thm:main}.
The argument uses only the triangle inequality; in particular, it does not invoke Ptolemy's inequality and is therefore valid on every metric space.
Fix a profile $\mathbf x$ with \(Z=\sum_{i<j}d(x_i,x_j)>0\), and fix any two comparison facilities \(f_A,f_B\).
Assign each agent to a nearer comparison facility, breaking ties arbitrarily, and let \(A,B\) be the resulting clusters.
Write
\[
 c_i=c_i(\{f_A,f_B\},\mathbf x)= \min \{d(x_i,f_A),d(x_i,f_B)\}.
\]
as the cost of agent $i$ under the comparison pair.
Set
\[
 \SC_A^*=\sum_{i\in A}c_i,\qquad
 \SC_B^*=\sum_{k\in B}c_k,\qquad
 \SC^*=\SC_A^*+\SC_B^*.
\]
Here \(\SC_A^*\) and \(\SC_B^*\) are the total service costs of the two clusters, and
\(
 \SC^*=\SC(\{f_A,f_B\},\mathbf{x})
\) is the social cost of the comparison pair \(f_A,f_B\).
The superscript \(*\) refers throughout this subsection to this fixed comparison pair.
The pair need not attain the optimum.
All inequalities below hold for every such pair, so we may take the infimum over pairs at the end.

The proof first bounds the social cost of Pairwise-Distance in terms of the comparison cost \(\SC^*\), from which a 4-approximation follows.

\begin{lemma}\label{lem:pd-accounting}
 {There are nonnegative quantities \(Q_j\), \(j\in N\), whose
values are specified in Definition~\ref{def:mixed-excess-assignment}, such that}
\begin{equation}\label{eq:global-pd}
 Z\bigl(\SC(\PD)-3\SC^*\bigr)\le\sum_{j\in N}Q_j
\end{equation}
In particular, \(\SC(\PD)\le4\OPT\).
\end{lemma}

\begin{proof}
If Pairwise-Distance selects the pair \(\{i,j\}\), every agent
\(k\notin\{i,j\}\) incurs cost \(\min\{d_{ik},d_{jk}\}\).
Since this pair is selected with probability \(d_{ij}/Z\),
\begin{align}
 Z\SC(\PD)
 &=\sum_{i<j}d_{ij}
   \sum_{k\in N\setminus\{i,j\}}\min\{d_{ik},d_{jk}\}\notag\\
 &=\sum_{i<j<k}\Bigl(
   d_{ij}\min\{d_{ik},d_{jk}\}
   +d_{ik}\min\{d_{ij},d_{jk}\}
   +d_{jk}\min\{d_{ij},d_{ik}\}
   \Bigr).
 \label{eq:compressed-triple-sum}
\end{align}
The second equality groups together the three terms involving the same three-element subset of agents.
For each \(i<j<k\), let \(S_{ijk}\) denote the corresponding summand.

For comparison with the baseline \(3Z\SC^*\),
we therefore associate with each triple $\{i,j,k\}$ a quantity \(\phi_{ijk}\), informally representing the portion of \(Z\SC^*\) assigned to that triple.
The definitions below are chosen so that
\(\sum_{i<j<k}3\phi_{ijk}\le 3Z\SC^*\).
We compare \(S_{ijk}\) with \(3\phi_{ijk}\) for each triple and assign any uncovered part to one of the quantities \(Q_j\).

A triple \(\{i,j,k\}\) is \emph{monochromatic} if all three agents belong to \(A\) or all three belong to \(B\); otherwise it is \emph{mixed}.
Because \(A,B\) partition \(N\), every mixed triple has exactly two agents in one cluster and one agent in the other.

Suppose first that \(\{i,j,k\}\) is monochromatic, and relabel its agents so that
\(d_{ij}\le d_{ik}\le d_{jk}\). In this case, \(S_{ijk}=d_{ij}(d_{ik}+d_{jk})\).
Define
\[
 \phi_{ijk}
 =\frac23\bigl(d_{jk}c_i+d_{ik}c_j+d_{ij}c_k\bigr).
\]
Because the three agents share a comparison facility, we have
 $c_i+c_j\ge d_{ij}$,
 $c_i+c_k\ge d_{ik}$ and
 $c_j+c_k\ge d_{jk}$.
Multiply these inequalities, respectively, by nonnegatove multipliers
 $d_{jk}+d_{ik}-d_{ij}$,
 $d_{jk}+d_{ij}-d_{ik}$,
 $d_{ij}+d_{ik}-d_{jk}$, 
and summing gives
\[
 3\phi_{ijk}
 \ge
 2(d_{ij}d_{ik}+d_{ij}d_{jk}+d_{ik}d_{jk})
 -(d_{ij}^2+d_{ik}^2+d_{jk}^2).
\]
The difference between the right-hand side and \(S_{ijk}\) is
\[
 d_{ij}(d_{ik}-d_{ij})
 +(d_{jk}-d_{ik})(d_{ij}+d_{ik}-d_{jk})\ge0.
\]
Hence, we have
\begin{equation}\label{eq:mono-three}
 S_{ijk}\le3\phi_{ijk}
\end{equation}
for every monochromatic triple.

Now suppose that \(\{i,j,k\}\) is mixed.
Relabel its agents so that \(i,j\) are the unique same-cluster pair and \(d_{ik}\le d_{jk}\).
Define
\[
 \phi_{ijk}=d_{jk}c_i+d_{ik}c_j.
\]
For later use, we show that every mixed triple satisfies
\begin{equation}\label{eq:mixed-four-compressed}
 S_{ijk}\le4\phi_{ijk}.
\end{equation}
Indeed, we have
\(
 \phi_{ijk}
 \ge d_{ik}(c_i+c_j)
 \ge d_{ik}d_{ij}
\).
If \(d_{ij}\le d_{ik}\le d_{jk}\), then
\(
 S_{ijk}=d_{ij}(2d_{ik}+d_{jk})
 \le4d_{ij}d_{ik}
\) by the triangle inequality.
If \(d_{ik}\le d_{ij}\le d_{jk}\), then
\(
 S_{ijk}=d_{ik}(2d_{ij}+d_{jk})
 \le4d_{ik}d_{ij}.
\)
Finally, if \(d_{ik}\le d_{jk}\le d_{ij}\), then
\(
 S_{ijk}=d_{ik}(2d_{jk}+d_{ij})
 \le3d_{ik}d_{ij}.
\)
This proves \eqref{eq:mixed-four-compressed}.

\begin{definition}[Mixed-triple assignment]\label{def:mixed-excess-assignment}
 {For each mixed triple \(\{i,j,k\}\) with $i,j$ in the same cluster, add
\(\pos{S_{ijk}-3\phi_{ijk}}\) to \(Q_j\), where \(j\) is the farther
same-cluster agent from \(k\) (breaking ties consistently).  Thus, \(Q_j\)
is the total excess assigned to agent \(j\), and every mixed-triple excess is
counted exactly once.}
\end{definition}

It remains to bound the total baseline contribution.
Fix an agent \(i\).
For each pair \(\{j,k\}\subseteq N\setminus\{i\}\), the coefficient of \(c_i\) in \(\phi_{ijk}\) is either \(0\), \(2d_{jk}/3\), or \(d_{jk}\), according to the position of \(i\) and the type of the triple.
Its total coefficient is therefore at most \(\sum_{j<k}d_{jk}=Z\).
Summing over agents gives
\begin{equation}\label{eq:phi-congestion}
 \sum_{i<j<k}\phi_{ijk}
 \le Z\sum_{i\in N}c_i
 =Z\SC^*.
\end{equation}
Equations \eqref{eq:mono-three}, \eqref{eq:compressed-triple-sum}, and \eqref{eq:phi-congestion} now give
\[
 Z\SC(\PD)=\sum_{i<j<k}S_{ijk}
 \le
 3\sum_{i<j<k}\phi_{ijk}+\sum_jQ_j
 \le
 3Z\SC^*+\sum_jQ_j,
\]
which proves \eqref{eq:global-pd}.

For the final assertion that $\PD$ is a $4$-approximation, \eqref{eq:mono-three} and \eqref{eq:mixed-four-compressed} give \(S_{ijk}\le4\phi_{ijk}\) for every triple.
Thus, for every comparison pair of social cost \(\SC^*\),
\[
 Z\SC(\PD)
 =\sum_{i<j<k}S_{ijk}
 \le4\sum_{i<j<k}\phi_{ijk}
 \le4Z\SC^*.
\]
Dividing by \(Z\) and taking the infimum over comparison pairs proves the final assertion.
\end{proof}

Next, we derive a corresponding bound for Proportional and then show, through a one-row inequality for each anchor, that the savings of Proportional control these residual quantities.

For every \(j\in N\), define
\[
 D_{jA}=\sum_{\ell\in A}d_{j\ell},\qquad
 D_{jB}=\sum_{k\in B}d_{jk},\qquad
 D_j=D_{jA}+D_{jB}.
\]
Thus, \(D_{jA}\) and \(D_{jB}\) are the total distances from anchor \(j\) to clusters \(A\) and \(B\), respectively, and \(D_j\) is its total distance to all agents.
Define
\[
 K_j=
 \begin{cases}
 D_{jA}^2+
 \displaystyle\sum_{\ell\in A}d_{j\ell}
 \sum_{r\in N}\pos{d_{jr}-d_{\ell r}}, & j\in A,\\[2mm]
 D_{jB}^2+
 \displaystyle\sum_{\ell\in B}d_{j\ell}
 \sum_{r\in N}\pos{d_{jr}-d_{\ell r}}, & j\in B.
 \end{cases}
\]
In each case, the double sum is the sum, weighted by \(d_{j\ell}\), of the reductions obtained by opening the second facility at \(x_\ell\), over agents \(\ell\) in the same cluster as the first selected agent \(j\).

\begin{lemma}\label{lem:pm-accounting}
The Proportional mechanism satisfies
\begin{equation}\label{eq:global-pm}
 n\bigl(4\SC^*-\SC(\PM)\bigr)
 \ge
 |B|\SC_A^*+|A|\SC_B^*
 +\sum_{j\in N}\frac{K_j}{D_j}.
\end{equation}
\end{lemma}

\begin{proof}
Fix an anchor \(j\in A\), and let \(\SC_j\) be Proportional's expected social cost conditional on choosing \(j\) as the anchor.
Conditional on the anchor \(j\), the second agent \(\ell\) is selected with probability \(d_{j\ell}/D_j\).
We consider separately whether \(\ell\) belongs to \(A\) or \(B\) to upper-bound \(D_j\SC_j\).

\paragraph{Case 1: \(\ell\in A\).}
For a fixed \(\ell\in A\), the resulting social cost is
\[
 \sum_{r\in N}\min\{d_{jr},d_{\ell r}\}
 =D_j-\sum_{r\in N}\pos{d_{jr}-d_{\ell r}}.
\]
Consequently, the part of \(\SC_j\) arising from outcomes in which the second selected agent \(\ell\) belongs to \(A\) is
\(
 \sum_{\ell\in A}\frac{d_{j\ell}}{D_j}
 \left(D_j-\sum_{r\in N}\pos{d_{jr}-d_{\ell r}}\right).
\)
After multiplying by \(D_j\), this contribution becomes
\begin{align*}
 &\sum_{\ell\in A}d_{j\ell}
 \left(D_j-\sum_{r\in N}\pos{d_{jr}-d_{\ell r}}\right)=
 D_{jA}D_j-
 \sum_{\ell\in A}d_{j\ell}
 \sum_{r\in N}\pos{d_{jr}-d_{\ell r}}=D_{jA}D_j-(K_j-D_{jA}^2),
\end{align*}
where the last equality follows from the definition of \(K_j\) for \(j\in A\).

\paragraph{Case 2: \(\ell\in B\).}
We now bound the part of \(D_j\SC_j\) arising from outcomes in which the second selected agent belongs to \(B\).
For every such outcome, the facility at \(x_j\) can serve cluster \(A\).
Therefore, the contribution of cluster \(A\), summed over all \(\ell\in B\), is at most
\[
 \sum_{\ell\in B}d_{j\ell}
 \sum_{r\in A}\min\{d_{jr},d_{\ell r}\}
 \le
 \sum_{\ell\in B}d_{j\ell}D_{jA}
 =D_{jA}D_{jB}.
\]
We next bound the contribution of cluster \(B\).
For every ordered pair \(\ell,r\in B\),
\begin{equation}\label{eq:opp}
 d_{j\ell}\min\{d_{jr},d_{\ell r}\}
 \le d_{j\ell}c_r+2c_\ell d_{jr}.
\end{equation}
To prove this inequality, first suppose that \(d_{j\ell}\le2d_{jr}\).
Since \(\ell\) and \(r\) are assigned to the same comparison facility,
\[
 d_{j\ell}\min\{d_{jr},d_{\ell r}\}
 \le d_{j\ell}d_{\ell r}
 \le d_{j\ell}(c_\ell+c_r)
 \le2c_\ell d_{jr}+d_{j\ell}c_r.
\]
Now suppose that \(d_{j\ell}>2d_{jr}\).
The triangle inequality gives
\(d_{\ell r}\ge d_{j\ell}-d_{jr}>d_{jr}\), so the left-hand side is \(d_{j\ell}d_{jr}\).
Moreover,
\(
 d_{j\ell}-d_{jr}\le d_{\ell r}\le c_\ell+c_r\) and
 \(d_{jr}<d_{\ell r}\le c_\ell+c_r.
\)
Hence, we obtain
\[
 d_{j\ell}d_{jr}
 =(d_{j\ell}-d_{jr})d_{jr}+d_{jr}^2
 \le2(c_\ell+c_r)d_{jr}
 \le2c_\ell d_{jr}+c_r d_{j\ell},
\]
as desired.
Summing the ordered-pair inequality \eqref{eq:opp} over \(\ell,r\in B\) yields
\begin{align*}
 \sum_{\ell\in B}d_{j\ell}
 \sum_{r\in B}\min\{d_{jr},d_{\ell r}\}
 &\le
 \sum_{\ell,r\in B}\bigl(d_{j\ell}c_r+2c_\ell d_{jr}\bigr)
 =3D_{jB}\SC_B^*.
\end{align*}
Thus, the part of \(D_j\SC_j\) arising from outcomes in which the second selected agent belongs to \(B\) is at most
\(D_{jA}D_{jB}+3D_{jB}\SC_B^*\).

Consequently, combining Case 1 and Case 2 gives
\begin{align*}
 D_j\SC_j
 &\le
 D_{jA}D_j-(K_j-D_{jA}^2)+D_{jA}D_{jB}
 +3D_{jB}\SC_B^*\\
 &=2D_{jA}D_j-K_j+3D_{jB}\SC_B^*\\
 &\le D_j(2D_{jA}+3\SC_B^*)-K_j.
\end{align*}
Thus, for the anchor $j\in A$, we obtain
\begin{equation}\label{eq:compressed-anchor}
\SC_j\le2D_{jA}+3\SC_B^*-\frac{K_j}{D_j}.
\end{equation}
The symmetric bound holds for anchors in \(B\).

Summing \eqref{eq:compressed-anchor} and its symmetric counterpart over all anchors 
gives
\[
 n\SC(\PM)=\sum_{j\in N}\SC_j
 \le
 2\sum_{j\in A}D_{jA}
 +2\sum_{j\in B}D_{jB}
 +3|A|\SC_B^*+3|B|\SC_A^*
 -\sum_{j\in N}\frac{K_j}{D_j}.
\]
Subtracting from \(4n\SC^*\) yields
\begin{align*}
 &n\bigl(4\SC^*-\SC(\PM)\bigr)\\
 \ge~&
 |B|\SC_A^*+|A|\SC_B^*+4\left(|A|\SC_A^*
 -\frac12\sum_{j\in A}D_{jA}\right)+4\left(|B|\SC_B^*
 -\frac12\sum_{j\in B}D_{jB}\right)
 +\sum_{j\in N}\frac{K_j}{D_j}.
\end{align*}
Finally,
\[
 \frac12\sum_{j\in A}D_{jA}
 =\sum_{\{i,j\}\subseteq A}d_{ij}
 \le\sum_{\{i,j\}\subseteq A}(c_i+c_j)
 =(|A|-1)\SC_A^*,
 \qquad
 \frac12\sum_{j\in B}D_{jB}
 \le (|B|-1)\SC_B^*,
\]
because \(d_{ij}\le c_i+c_j\) for agents in the same cluster.
The two parenthesized terms are therefore nonnegative. Discarding them proves \eqref{eq:global-pm}.
\end{proof}

The remaining task is local.
For an anchor \(j\), the next lemma gives a lower bound on its corresponding Proportional savings in the right-hand side of \eqref{eq:global-pm}.
In particular, the savings are sufficient to cover
\((4\sqrt3-6)n/Z\) times the mixed-triple excess \(Q_j\) assigned to \(j\) in Definition~\ref{def:mixed-excess-assignment}.

\begin{lemma}\label{lem:row-charging}
For every \(j\in A\),
\begin{equation}\label{eq:row-charging}
 |B|c_j+\frac{K_j}{D_j}
 \ge
(4\sqrt3-6)n\frac{Q_j}{Z}.
\end{equation}
The symmetric inequality holds for \(j\in B\), with \(|A|\) in place of \(|B|\).
\end{lemma}

\begin{proof}
The claim is immediate if \(Q_j=0\).
Assume \(Q_j>0\).
Then \(|B|>0\) due to the existence of mixed triples.

Consider an assigned mixed triple with \(i,j\in A\) and \(k\in B\), oriented
so that \(d_{ik}\le d_{jk}\).  Recall that $S_{ijk}=d_{ij}\min\{d_{ik},d_{jk}\}
   +d_{ik}\min\{d_{ij},d_{jk}\}
   +d_{jk}\min\{d_{ij},d_{ik}\}$ and
\(\phi_{ijk}=d_{jk}c_i+d_{ik}c_j\). Define
\[
 z_{ji}=\sqrt{d_{ij}(d_{ij}-3c_i)_+}.
\]
The three possible positions of \(d_{ij}\) relative to \(d_{ik}\le d_{jk}\),
together with \(d_{ij}\le c_i+c_j\), give
\[
 S_{ijk}-3\phi_{ijk}
 \le
 \begin{cases}
  0,
    & d_{ik}\le d_{jk}\le d_{ij},\\
  (d_{jk}-d_{ik})(d_{ij}-3c_i),
    & d_{ij}\le d_{ik}\le d_{jk},\\
  (d_{jk}-d_{ij})(d_{ik}-3c_i),
    & d_{ik}\le d_{ij}\le d_{jk}.
 \end{cases}
\]

The triangle inequality gives
\(
 d_{jk}-d_{ik}\le d_{ij}\) and
 \(d_{jk}-d_{ij}\le d_{ik}.
\)
In the third case, we also have
\(d_{jk}-d_{ij}\le d_{jk}-d_{ik}\).
These inequalities imply the simultaneous bounds
\begin{equation}\label{eq:three-local-compressed}
 \pos{S_{ijk}-3\phi_{ijk}}\le z_{ji}^2,\qquad
 \pos{S_{ijk}-3\phi_{ijk}}\le z_{ji}d_{ik},\qquad
 \pos{S_{ijk}-3\phi_{ijk}}\le z_{ji}(d_{jk}-d_{ik}).
\end{equation}
Moreover, note that
\begin{equation}\label{eq:z-range-compressed}
 0\le z_{ji}\le d_{ij},\qquad z_{ji}\le c_j,
\end{equation}
where the last inequality is because 
\(d_{ij}\le c_i+c_j\), and
\(
 z_{ji}^2
 \le
 (c_i+c_j)(c_j-2c_i)
 \le c_j^2.
\)

Define
\begin{align*}
\pi_j=\sum_{i\in A\setminus\{j\}}z_{ji},\qquad G_j&=
 \sum_{\substack{i\in A\setminus\{j\}\\k\in B}}
 z_{ji}d_{ik},\qquad
 H_j=
 \sum_{\substack{i\in A\setminus\{j\}\\k\in B}}
 z_{ji}\pos{d_{jk}-d_{ik}}.
\end{align*}
Summing \eqref{eq:three-local-compressed} over the assigned triples and using \eqref{eq:z-range-compressed} gives
\begin{equation}\label{eq:compressed-row-constraints}
 Q_j\le G_j,\qquad
 Q_j\le H_j,\qquad
 Q_j\le |B|c_j\pi_j.
\end{equation}
Furthermore, because
\(d_{jk}\le d_{ik}+\pos{d_{jk}-d_{ik}}\) for every \(i,k\), we have
\begin{equation}\label{eq:compressed-piD}
 \pi_jD_{jB}\le G_j+H_j.
\end{equation}

For this fixed anchor \(j\in A\), define
\[
 R_j=
 \sum_{i\in A}d_{ji}
 \sum_{k\in N}\pos{d_{jk}-d_{ik}},
\]
so that \(K_j=D_{jA}^2+R_j\).
We next identify the savings contained in \(R_j\).
Split \(R_j=R_j^A+R_j^B\) according to whether the improved agent \(k\) belongs to \(A\) or \(B\).
Since \(d_{ij}\ge z_{ji}\),
\begin{equation}\label{eq:compressed-cross-saving}
 R_j^B\ge H_j.
\end{equation}

\begin{claim}\label{clm:compressed-energy}
\(R_j^A\ge\pi_j^2\).
\end{claim}

\begin{proof}
Call \(i\in A\setminus\{j\}\) \emph{active} if \(d_{ij}>3c_i\).
For each active \(i\), the diagonal term in \(R_j^A\) with $k=i$ is \(d_{ij}^2\ge z_{ji}^2\).
Now fix two distinct active agents \(i,k\in A\).  We first show
\begin{equation}\label{eq:compressed-active-geometry}
 c_i+c_k<\min\{d_{ij},d_{kj}\},\qquad
 d_{ij}<2d_{kj},\qquad d_{kj}<2d_{ij}.
\end{equation}
If \(d_{ij}\le c_i+c_k\), then \(d_{ij}>3c_i\) implies
\(c_k>2c_i\), and
\(
 d_{kj}\le d_{ij}+d_{ik}
 \le2(c_i+c_k)<3c_k,
\)
contradicting the activity of \(k\).
Thus \(d_{ij}>c_i+c_k\), and symmetrically
\(d_{kj}>c_i+c_k\).
The remaining two inequalities follow from
\(|d_{ij}-d_{kj}|\le d_{ik}\le c_i+c_k
<\min\{d_{ij},d_{kj}\}\).

The two cross terms in \(R_j^A\) generated by \(i,k\) are 
\[
d_{ji}\pos{d_{jk}-d_{ik}}+d_{jk}\pos{d_{ji}-d_{ki}}\ge
 2d_{ij}d_{kj}-(d_{ij}+d_{kj})(c_i+c_k).
\]
By \eqref{eq:compressed-active-geometry},
\[
 (c_i+c_k)\left(\frac1{d_{ij}}+\frac1{d_{kj}}\right)
 \le
 3\left(\frac{c_i}{d_{ij}}+\frac{c_k}{d_{kj}}\right).
\]
Indeed, after multiplication by \(d_{ij}d_{kj}\), the difference between
the RHS and LHS is
\(
 c_i(2d_{kj}-d_{ij})+c_k(2d_{ij}-d_{kj})\ge0.
\)
Since \(c_i/d_{ij},c_k/d_{kj}<1/3\),
\[
 \left(1-\frac32\left(\frac{c_i}{d_{ij}}+
       \frac{c_k}{d_{kj}}\right)\right)^2
 -\left(1-\frac{3c_i}{d_{ij}}\right)
  \left(1-\frac{3c_k}{d_{kj}}\right)
 =\frac94\left(\frac{c_i}{d_{ij}}-
       \frac{c_k}{d_{kj}}\right)^2\ge 0,
\]
and the first term inside the square is nonnegative.
Hence the two cross terms are at least
\[
 2d_{ij}d_{kj}
 \sqrt{\left(1-\frac{3c_i}{d_{ij}}\right)
       \left(1-\frac{3c_k}{d_{kj}}\right)}
 =2z_{ji}z_{jk}.
\]
Summing all diagonal and cross contributions proves the claim.
\end{proof}

Equation \eqref{eq:compressed-cross-saving} and Claim~\ref{clm:compressed-energy}, together with
\(\pi_j\le D_{jA}\) and \(H_j\le\pi_jD_{jB}\), imply
\begin{align}
 \frac{K_j}{D_j}
 &=
 \frac{D_{jA}^2+R_j}{D_{jA}+D_{jB}}\ge
 \frac{D_{jA}^2+\pi_j^2+H_j}{D_{jA}+D_{jB}}\ge
 \frac{2\pi_j^2+H_j}{\pi_j+D_{jB}},
 \label{eq:compressed-row-budget}
\end{align}
where the last step is because the function
\(
 u\to\frac{u^2+\pi_j^2+H_j}{u+D_{jB}}
\)
has non-negative derivatives when $u\ge\pi_j$.

We now compare the local scale \(G_j\) with the global denominator \(Z/n\).
The assumption \(Q_j>0\) and \eqref{eq:compressed-row-constraints} imply that
\(c_j,\pi_j,G_j\) are positive.
For \(0\le\theta\le1\), define
\[
 T_\theta=\{i\in A\setminus\{j\}:z_{ji}\ge\theta c_j\}.
\]
Then
\begin{equation*}
 G_j=c_j\int_0^1D(T_\theta,B)\,d\theta,
 \qquad
 \pi_j=c_j\int_0^1|T_\theta|\,d\theta,
\end{equation*}
where \(D(R,B)=\sum_{i\in R,k\in B}d_{ik}\).

\begin{claim}\label{clm:tyi}
$nG_j\le(\pi_j+|B|c_j)Z.$
\end{claim}
\begin{proof}
For any disjoint \(R,B\subseteq N\), we show that
$nD(R,B)\le(|R|+|B|)Z.$
Indeed, for every \(o\in N\setminus(R\cup B)\), the triangle inequality gives
\[
\begin{aligned}
 D(R,B)
 &\le\sum_{\substack{i\in R,k\in B}}(d_{io}+d_{ok})=|B|\sum_{i\in R}d_{io}
   +|R|\sum_{k\in B}d_{ko}.
\end{aligned}
\]
Summing over \(o\in N\setminus(R\cup B)\) yields
\[
\begin{aligned}
 |N\setminus(R\cup B)|D(R,B)
 &\le
 |B|\sum_{\substack{i\in R\\o\in N\setminus(R\cup B)}}d_{io}
 +|R|\sum_{\substack{k\in B\\o\in N\setminus(R\cup B)}}d_{ko}\le(|R|+|B|)\bigl(Z-D(R,B)\bigr).
\end{aligned}
\]
Since \(R\cap B=\varnothing\), we have
\(|N\setminus(R\cup B)|=n-|R|-|B|\), and the last inequality rearranges to $nD(R,B)\le(|R|+|B|)Z.$
Applying this inequality to every level set in the definition of $G_j$, we  obtain the claim.
\end{proof}

It remains to show that the local savings control the assigned excess \(Q_j\).
The next claim combines the row budget with the constraints on
\(G_j,H_j\), and \(Q_j\) to give the required comparison.
\begin{claim}\label{clm:compressed-scalar}
For the fixed anchor \(j\in A\),
\[
 |B|c_j+\frac{K_j}{D_j}
 \ge
 (4\sqrt3-6)(\pi_j+|B|c_j)\frac{Q_j}{G_j}.
\]
\end{claim}

\begin{proof}
By \eqref{eq:compressed-row-budget} and \eqref{eq:compressed-piD},
\[
 |B|c_j+\frac{K_j}{D_j}
 \ge
 |B|c_j+
 \frac{2\pi_j^2+H_j}{\pi_j+(G_j+H_j)/\pi_j}.
\]
Set the following three dimensionless ratios:
\[
 \beta=\frac{|B|c_j}{\pi_j},\qquad
 g=\frac{G_j}{|B|c_j\pi_j},\qquad
 y=\frac{H_j}{|B|c_j\pi_j}.
\]
Dividing the preceding lower bound by \(\pi_j\) gives
\[
 \frac{1}{\pi_j}\left(|B|c_j+\frac{K_j}{D_j}\right)
 \ge
 \beta+\frac{2+\beta y}{1+\beta(g+y)}.
\]
Moreover, \eqref{eq:compressed-row-constraints} implies
\(
 \frac{Q_j}{G_j}
 \le \min\left\{1,\frac{y}{g},\frac1g\right\}.
\)
It therefore suffices to prove
\[
 \frac{g}{(1+\beta)\min\{g,y,1\}}
 \left(
  \beta+\frac{2+\beta y}{1+\beta(g+y)}
 \right)
 \ge 4\sqrt3-6.
\]
We distinguish which of \(g,y,1\) attains the minimum, breaking ties
arbitrarily.

Suppose first that \(g\le\min\{y,1\}\).  If \(\beta g\le1\), then
\((2+\beta y)/(1+\beta(g+y))\ge1\), so the left-hand side is at least one.
If \(\beta g>1\), this fraction is increasing in \(y\), since its derivative is
\(
 \frac{\beta(\beta g-1)}{(1+\beta(g+y))^2}>0.
\)
Hence its minimum is attained at \(y=g\).

Next suppose that \(y\le\min\{g,1\}\).  For fixed \(\beta,y\), the derivative
with respect to \(g\) of
\(
 g\bigl(\beta+(2+\beta y)/(1+\beta(g+y))\bigr)
\)
is
\(
 \beta+
 \frac{(2+\beta y)(1+\beta y)}{(1+\beta(g+y))^2}>0.
\)
Thus the minimum is attained at \(g=y\).  In both of the preceding cases,
the value at \(g=y\) is at least one when \(\beta g\le1\).  We have therefore
reduced the nontrivial regime to \(g=y\le1\) and \(\beta g>1\), where
the expression becomes
\[
 \frac{1}{1+\beta}
 \left(\beta+\frac{2+\beta g}{1+2\beta g}\right).
\]
The second fraction depends only on \(\beta g\) and is smaller than one.
Holding \(\beta g\) fixed, the last display increases with \(\beta\).
Since \(g\le1\), its minimum is therefore attained at \(g=1\).

Finally, suppose that \(1\le\min\{g,y\}\).  The same derivative calculation
shows that the minimum is attained at \(g=1\).  If \(\beta\le1\), then
\((2+\beta y)/(1+\beta(1+y))\ge1\), and the expression is at least one.
If \(\beta>1\), this fraction is increasing in \(y\), so its minimum is
attained at \(y=1\).

All nontrivial cases have now reduced to minimizing
\(
 \frac{2(\beta^2+\beta+1)}{(\beta+1)(2\beta+1)}
\) for $\beta>1$.
By computing the derivative, the minimum is attained at \(\beta=1+\sqrt3\), where the value is
\(4\sqrt3-6\).  This proves the claim.
\end{proof}

Finally, Claim~\ref{clm:compressed-scalar} and
\ref{clm:tyi} give
\[
 |B|c_j+\frac{K_j}{D_j}
 \ge
(4\sqrt3-6)n\frac{Q_j}{Z},
\]
which proves the lemma.
\end{proof}

We now combine Lemmas~\ref{lem:pd-accounting}-\ref{lem:row-charging} to prove the upper bound
\(\SC(\HD)\le(3+\lambda^*)\OPT\) in Theorem~\ref{thm:main}.

\begin{proof}[Proof of Theorem~\ref{thm:main}: approximation upper bound]
We first establish the following tradeoff, valid on every metric space:
\begin{equation}\label{eq:joint-tradeoff}
\SC(\PM)+(4\sqrt3-6)\SC(\PD)
 \le \bigl(4+3(4\sqrt3-6)\bigr)\OPT.
\end{equation}
Summing Lemma~\ref{lem:row-charging} over all \(j\in N\) and applying
Lemma~\ref{lem:pm-accounting} yields
\[
 n\bigl(4\SC^*-\SC(\PM)\bigr)
 \ge
 (4\sqrt3-6)n\frac{\sum_{j\in N}Q_j}{Z}.
\]
Moreover, Lemma~\ref{lem:pd-accounting} gives
\(
 \frac{\sum_{j\in N}Q_j}{Z}
 \ge \SC(\PD)-3\SC^*.
\)
Combining these inequalities shows that every comparison pair of cost
\(\SC^*\) satisfies
\[
 \SC(\PM)+(4\sqrt3-6)\SC(\PD)
 \le \bigl(4+3(4\sqrt3-6)\bigr)\SC^*.
\]
Taking the infimum over comparison pairs proves
\eqref{eq:joint-tradeoff}.  

By the definition of \(\HD\) and \eqref{eq:joint-tradeoff}, we obtain
\begin{align}
 \SC(\HD)
 &=\frac{1}{4\sqrt3-5}\SC(\PM)
   +\frac{4\sqrt3-6}{4\sqrt3-5}\SC(\PD)\notag\\
 &\le
 \frac{4+3(4\sqrt3-6)}{4\sqrt3-5}\OPT\notag\\
 &=(3+\lambda^*)\OPT.
 \label{eq:hd-upper}
\end{align}
\end{proof}

\subsection{Tightness and the optimal mixing weight}
\label{sec:tightness}
We now show that the weight $\lambda^*=\frac{5+4\sqrt3}{23}$ in our mechanism $\HD=\Mix_{\lambda^*}$ is optimal among the mechanism class $\Mix_{\lambda}=\lambda\PM+(1-\lambda)\PD$, for all probabilities $\lambda\in[0,1]$ assigned to Proportional. 
We construct two families of line instances that determine the lower bound for the mixtures $\Mix_{\lambda}$.
On the first family, the approximation ratios of Proportional and Pairwise-Distance converge to \(4\) and \(3\), respectively.
On the second, the ratios of the two mechanisms converge to $10-4\sqrt3$ and \(4\), respectively.

\begin{proposition}\label{prop:lower-frontier}
For every \(\lambda\in[0,1]\),
\begin{equation}\label{eq:lower-frontier}
 \operatorname{apx}(\Mix_\lambda)
 \ge
 \max\{3+\lambda,\ 4-(4\sqrt3-6)\lambda\}.
\end{equation}
\end{proposition}

\begin{proof}
We analyze two families of profiles separately.

For the first family, let \(\eps=m^{-2}\) and consider the line profile
\(
 \mathbf{x}^m
 =(0^m,\eps,1),
\)
where $0^m$ means that $m$ agents are located at $0$. 
Its optimum is \(\eps\). Let $m\to \infty$. 
Directly conditioning on the Proportional anchor and enumerating Pairwise-Distance's three nonzero pair types gives
\begin{align*}
 \frac{\SC(\PM(\mathbf{x}^m),\mathbf{x}^m)}
 {\OPT(\mathbf{x}^m)}
 &=
 \frac{m}{m+2}
 \left[
 \frac{2-\eps}{1+\eps}
 +\frac{2(1-\eps)}{m\eps+1-\eps}
 +\frac{2-\eps}{m+1-\eps}
 \right]
 \longrightarrow4,\\
 \frac{\SC(\PD(\mathbf{x}^m),\mathbf{x}^m)}
 {\OPT(\mathbf{x}^m)}
 &=
 \frac{m(3-2\eps)}
 {m(1+\eps)+1-\eps}
 \longrightarrow3.
\end{align*}
Therefore, we have
\(
 \operatorname{apx}(\Mix_\lambda)\ge3+\lambda.
\)

For the second family, take \(A\) agents at \(0\), \(B\) agents at \(1\), and one agent at \(2\):
\(
 \mathbf{x}^{A,B}
 =(0^A,1^B,2).
\)
Its optimum is \(1\). Let \(A,B\to\infty\) with \(A/B\to t>0\).
By enumeration we have
\begin{align*}
 \SC(\PD(\mathbf{x}^{A,B}),\mathbf{x}^{A,B})
 &=
 \frac{4AB}{AB+2A+B}\longrightarrow4,\\
 \SC(\PM(\mathbf{x}^{A,B}),\mathbf{x}^{A,B})
 &=
 \frac{1}{A+B+1}
 \left[
 \frac{3AB}{B+2}
 +\frac{2AB}{A+1}
 +\frac{3AB}{2A+B}
 \right]\longrightarrow
 \psi(t)
 :=
 \frac{6t^2+10t+2}{2t^2+3t+1}.
\end{align*}
It is easy to verify that the function \(\psi\) is maximized at
\(
 \psi(1+\sqrt3)=10-4\sqrt3.
\)
Along this sequence,
\[
 \frac{\SC(\Mix_\lambda)}{\OPT}
 \longrightarrow
 \lambda\bigl(10-4\sqrt3\bigr)+(1-\lambda)4
 =4-(4\sqrt3-6)\lambda.
\]
Combining the two families proves the proposition.
\end{proof}

Taking \(\lambda=0\) in Proposition~\ref{prop:lower-frontier} and combining the result with Lemma~\ref{lem:pd-accounting} shows that Pairwise-Distance has approximation ratio exactly \(4\), even on the line.

Now we complete the proof of Theorem~\ref{thm:main}.

\begin{proof}[Proof of Theorem~\ref{thm:main}: Optimal mixture]
The two numbers in \eqref{eq:lower-frontier} meet at
\(
 3+\lambda=4-(4\sqrt3-6)\lambda,
\)
whose unique solution is
\(\lambda^*=1/(4\sqrt3-5)\).
Their common value is
\(
 3+\lambda^*.
\)
Thus every fixed mixture has approximation ratio at least \(3+\lambda^*\).
Moreover, Equation~\eqref{eq:hd-upper} gives the matching upper bound for \(\HD=\Mix_{\lambda^*}\).
\end{proof}

\section{Multiple Facilities: \texorpdfstring{$k=n-1$}{k=n-1}}\label{sec:k=n-1}

For locating \(k=n-1\) facilities in a metric space,
Escoffier et al.~\cite{escoffier2011many} introduced the \emph{Inversely Proportional} mechanism, 
which omits agent \(i\) with probability proportional to the reciprocal of her nearest-neighbor distance. It is strategyproof and has an \(n/2\)-approximation.
We introduce the \emph{Inverse-Square} mechanism by raising the reciprocal to the second power, and we show that it preserves strategyproofness and improves the approximation ratio to \(\Theta(\sqrt n)\).

For an all-distinct report profile \(\mathbf x\), let
\[
 d_i(\mathbf x)=\min_{j\ne i}d(x_i,x_j)
\]
be the minimum distance from agent \(i\) to any other agent, and write the shorthand $d_i$ when the context is clear.
The \emph{Inverse-Square mechanism} \(\mathsf{IS}\) omits agent \(i\) with probability
\[
 p_i(\mathbf x)=
 \frac{d_i(\mathbf x)^{-2}}{\sum_{j=1}^n d_j(\mathbf x)^{-2}}
\]
and opens facilities at all other reports.
If there are at most \(n-1\) distinct reports, it covers every distinct reported position.

\begin{theorem}\label{thm:inverse-square}
For every metric space and every \(n\ge3\), the Inverse-Square mechanism for \(k=n-1\) facilities is strategyproof and has an approximation ratio of
\(
 \frac{1+\sqrt{n/2}}2.
\)
\end{theorem}

\begin{proof}
  We first prove the strategyproofness.
If a profile has a repeated position, every occupied position is covered and all agents have zero cost.
Thus, no agent can benefit from misreporting.
We may therefore fix an agent \(i\) in an all-distinct profile $\mathbf x$ and normalize her nearest-neighbor distance to \(1\).
Write
\[
 O=\sum_{j\ne i}\frac1{d_j^2(\mathbf x)}.
\]
The mechanism omits \(i\) with probability \(1/(1+O)\), and she then pays her nearest-neighbor distance \(1\).
Hence her truthful expected cost is \(1/(1+O)\).
If \(h\) is a truthful nearest neighbor of \(i\), then \(d_h(\mathbf x)\le1\), so \(O\ge1\) and the truthful cost is at most \(1/2\).

Suppose that \(i\) reports \(x_i'\).
Every outcome then costs her at least \(\min\{d(x_i,x_i'),1\}\).
Therefore, a deviation with \(d(x_i,x_i')\ge1/(1+O)\) cannot help.
Reporting another agent's position is also unprofitable: that report is at distance at least \(1\) from \(x_i\), whereas truthful reporting costs at most \(1/2\).
Consider an all-distinct deviating profile satisfying
\begin{equation*}
 0<d(x_i,x_i')<\frac1{1+O}\le\frac12.
\end{equation*}

Let
\(
 d_i'=\min_{j\ne i}d(x_i',x_j)
\)
be the nearest-neighbor distance of the deviating report.
The triangle inequality gives
\[
 1-d(x_i,x_i')\le d_i'\le1+d(x_i,x_i').
\]
Let \(\mathcal R\) be the set of other agents whose nearest-neighbor distance decreases after the deviation.
For every \(j\in\mathcal R\), let \(y_j\) and \(y'_j\) be her inverse-square weights before and after the deviation, respectively.
Then
\[
 y'_j=\frac1{d(x_j,x_i')^2},
 \qquad
 y_j\ge\frac1{(d(x_j,x_i')+d(x_i,x_i'))^2},
 \qquad
 d(x_j,x_i')\ge d_i'.
\]

For agents outside \(\mathcal R\), the inverse-square weight does not increase.
Let
\[
 B=\sum_{j\in N\setminus(\mathcal R\cup\{i\})}
 \frac1{d_j^2(\mathbf x)}
\]
be the sum of their old weights. Thus $O=B+\sum_{j\in\mathcal R}y_j$. Let \(O'=B+\sum_{j\in\mathcal R}y'_j\).
Thus, \(O'\) uses the new weights \(y'_j\) for the agents in \(\mathcal R\), whose weights increase, and the old weights for all other agents.
Because the weights of agents outside \(\mathcal R\) do not increase, \(O'\) is at least the actual sum of the post-deviation weights of all agents other than \(i\).
If the deviating report $x_i'$ is selected, agent \(i\) pays \(d(x_i,x_i')\); if it is omitted, she pays \(1\).
Her expected cost after deviating is therefore at least
\[
 d(x_i,x_i')+\bigl(1-d(x_i,x_i')\bigr)
 \frac{(d_i')^{-2}}{(d_i')^{-2}+O'}=
\frac{(d_i')^{-2}+d(x_i,x_i')O'}
{(d_i')^{-2}+O'}.
\]
After comparing this expression with \(1/(1+O)\) and clearing the positive denominators, it is enough to prove
\begin{equation}\label{eq:is-phi}
 \frac{O}{(d_i')^2}-O'+d(x_i,x_i')O'(1+O)\ge0.
\end{equation}

We use the following elementary estimate for every \(j\in\mathcal R\):
\begin{equation}\label{eq:is-local-lemma}
 \frac{y_j}{(d_i')^2}-y'_j+d(x_i,x_i')
 \left(
  2y'_j+\frac{y_j}{(1+d(x_i,x_i'))^2}+y_jy'_j
 \right)
 \ge0.
\end{equation}
Indeed, the left-hand side is increasing in \(y_j\), so we may replace \(y_j\) by
\(1/(d(x_j,x_i')+d(x_i,x_i'))^2\).
The remaining inequality follows by clearing denominators and considering whether
\(d(x_j,x_i')\) is at most or at least \(1+d(x_i,x_i')\); in both cases the numerator is nonnegative for
\(0<d(x_i,x_i')<1/2\), using
\(1-d(x_i,x_i')\le d_i'\le1+d(x_i,x_i')\) and \(d(x_j,x_i')\ge d_i'\).

We now verify \eqref{eq:is-phi}, using the truthful nearest neighbor \(h\) of \(i\).
First suppose that \(h\notin\mathcal R\).
Then \(B\ge1\). Starting from an empty \(\mathcal R\), the left-hand side of \eqref{eq:is-phi} is 
\begin{align*}
\frac{B}{(d_i')^2}-B+d(x_i,x_i')B(1+B)
&\ge 
 B\left[
  \frac1{(1+d(x_i,x_i'))^2}-1+d(x_i,x_i')(1+B)
 \right]\\
 &\ge
 B\frac{d(x_i,x_i')^2(3+2d(x_i,x_i'))}{(1+d(x_i,x_i'))^2}
 \ge0.
\end{align*}
Add the agents in \(\mathcal R\) one at a time.
At every step, the accumulated old weight is at least \(1\), and the accumulated new weight is at least
\(1/(1+d(x_i,x_i'))^2\).
Consequently, the increase in the left-hand side of \eqref{eq:is-phi} is bounded below by the expression in \eqref{eq:is-local-lemma}, and is therefore nonnegative.

Second, suppose that \(h\in\mathcal R\), and add \(h\) first.
Its weights satisfy
\[
 y_h\ge1,
 \qquad
 \frac1{(1+d(x_i,x_i'))^2}\le y'_h\le\frac1{(d_i')^2}.
\]
Using these bounds directly, the left-hand side of \eqref{eq:is-phi} after adding \(h\) is at least
\[
\begin{aligned}
 &B\left(\frac1{(1+d(x_i,x_i'))^2}-1\right)
 +d(x_i,x_i')\left(B+\frac1{(1+d(x_i,x_i'))^2}\right)(B+2)\\
 &\quad=
 \frac{d(x_i,x_i')}{(1+d(x_i,x_i'))^2}
 \left[
  B^2(1+d(x_i,x_i'))^2
  +B(1+3d(x_i,x_i')+2d(x_i,x_i')^2)+2
 \right]
 \ge0.
\end{aligned}
\]
The accumulated old and new weights now satisfy the same two lower bounds as above, so adding each remaining agent in \(\mathcal R\) cannot decrease the expression, by \eqref{eq:is-local-lemma}.
This proves \eqref{eq:is-phi}, and hence strategyproofness.

It remains to determine the approximation ratio.
Let \(d_*=\min_i d_i\) be the minimum distance between any two agents.
With \(n-1\) facilities, at least two agents share a facility.  Serving a closest pair with one facility and every other agent at her own position gives \(\OPT_{n-1}= d_*\). 
The expected social cost is
\(
 \frac{\sum_i d_i^{-1}}{\sum_i d_i^{-2}}.
\)
At least two agents of a closest pair have \(d_i=d_*\).
For every other agent $j$, \(d_*/d_j\in[0,1]\). Therefore, 
\[
 \frac{\SC(\mathsf{IS})}{\OPT_{n-1}}
 =\frac{2+\sum_j d_*/d_j}{2+\sum_j(d_*/d_j)^2}
 \le
 \frac{2+\sum_j d_*/d_j}
 {2+\left(\sum_jd_*/d_j\right)^2/(n-2)}
 \le\frac{2+\sqrt{2n}}4
 =\frac{1+\sqrt{n/2}}2,
\]
where both sums range over the \(n-2\) agents outside the closest pair. The first inequality is Cauchy--Schwarz, and for the second inequality, the function
\(
 x\longmapsto\frac{2+x}{2+x^2/(n-2)}
\)
attains its maximum over \(x\in[0,n-2]\) at
\(x=\sqrt{2n}-2\). This proves the approximation ratio. 

For tightness, consider a line profile with \(x_1=0\), \(x_2=1\), and
\(x_j=1+(j-2)(1+\sqrt{n/2})\) for \(j\ge3\).
Then \(d_1=d_2=1\), every other nearest-neighbor distance is \(1+\sqrt{n/2}\), and direct substitution attains the bound exactly.
\end{proof}

\section{Multiple Facilities: \texorpdfstring{$k=3$}{k=3} on the Line}
\label{sec:k-facility}


In this section, we consider the problem of locating $k=3$ facilities on the real line.
We propose a randomized strategyproof mechanism called \emph{Gap-Product} and prove that it achieves an approximation ratio of $6$.
The previous best-known result for this setting is due to Fotakis and Tzamos \cite{fotakis2013strategyproof}: they proposed a randomized group-strategyproof mechanism, called \emph{Equal Cost}, that achieves an $n$-approximation on the line for any $k$ and $n$.

Let $\mathbf x\in\mathbb R^n$ be a reported profile on the line with \(x_1\le\cdots\le x_n\).
For an ordered \(r\)-agent subset \(S=\{s_1<\cdots<s_r\}\subseteq N\), define
\[
 w(S)=\prod_{j=1}^{r-1}(x_{s_{j+1}}-x_{s_j}),
 \qquad
 Z_r=\sum_{\substack{S\subseteq N\\|S|=r}}w(S).
\]
As usual, the empty product gives \(w(S)=1\) when \(S\) is a singleton.
We define the mechanism for general \(k\ge2\): the \emph{Gap-Product} mechanism \(\mathsf{GPM}_k\) selects a \(k\)-agent subset \(S\) with probability \(w(S)/Z_k\) and opens facilities at the selected reports.
If the profile has at most \(k\) distinct occupied positions, the mechanism instead covers every occupied position.
For \(k=2\), \(w(\{i,j\})=|x_i-x_j|\), so \(\mathsf{GPM}_2\) is exactly Pairwise-Distance.

\begin{theorem}\label{thm:gpm-three}
    For $k=3$ facilities on the line, the \emph{Gap-Product} mechanism is strategyproof and achieves a $6$-approximation. 
\end{theorem}

We first prove the approximation guarantee and then prove strategyproofness. While the exact approximation ratio is $2k$ for every $k\ge2$, strategyproofness fails for every $k\ge4$.

\begin{lemma}[Approximation]\label{lem:gpm-approximation}
For every line profile and every \(k\ge2\),
\[
 \SC(\mathsf{GPM}_k)\le 2k\frac{Z_{k+1}}{Z_k}\le2k\OPT_k.
\]
The factor \(2k\) is asymptotically tight for this mechanism.
\end{lemma}

\begin{proof}
If the profile has at most $k$ occupied locations, the mechanism has zero cost, so assume
otherwise. 

We first give an equivalent description of the mechanism.
For every \(t\in\{1,\ldots,n-1\}\), regard the interval between \(x_t\) and \(x_{t+1}\) as a possible cut.
Choose \(k-1\) cuts, which divide the ordered profile into \(k\) nonempty consecutive blocks.
For every such partition, give a weight
\[
 \left(
  \prod_{\text{Cut }\substack{t}}
  (x_{t+1}-x_t)
 \right)
 \left(\prod_{\text{Block }I}|I|\right).
\]
After choosing the partition proportionally to this weight, choose one agent uniformly from each block and open a facility at each selected report.

To verify that this is Gap-Product, fix
\(S=\{s_1<\cdots<s_k\}\).
For the selected reports to consist of \(S\), there must be one cut between \(s_j\) and \(s_{j+1}\) for every \(j\in\{1,\ldots,k-1\}\).
Moreover,
\[
 x_{s_{j+1}}-x_{s_j}
 =\sum_{t=s_j}^{s_{j+1}-1}(x_{t+1}-x_t).
\]
Thus, after summing over all compatible choices of cuts, the total weight assigned to \(S\) is
\[
 \prod_{j=1}^{k-1}(x_{s_{j+1}}-x_{s_j})=w(S).
\]
The product of the block sizes in the partition weight is canceled by the uniform choice of one agent from each block.
Consequently, the total weight of all \(k\)-block partitions is \(Z_k\), and the alternative rule selects every \(S\) with probability \(w(S)/Z_k\).

We next bound the social cost in this equivalent rule.
Assign every agent to the facility selected from her own block.
This assignment may not use the nearest selected facility, so its expected cost is an upper bound on \(\SC(\mathsf{GPM}_k)\).
For a fixed block \(I\), summing this assigned cost over all possible choices of its selected agent gives
\[
 \sum_{r\in I}\sum_{i\in I}|x_i-x_r|
 =
 2\sum_{\substack{t:t,t+1\in I}}
 |\{i\in I:i\le t\}|\,
 |\{i\in I:i>t\}|\,
 (x_{t+1}-x_t).
\]
Indeed, an interval between \(x_t\) and \(x_{t+1}\) contributes to \(|x_i-x_r|\) exactly when \(i\) and \(r\) lie on opposite sides of that interval, and the factor \(2\) accounts for the two possible orders of \(i\) and \(r\).

Consider one term in this sum and add a cut between \(x_t\) and \(x_{t+1}\).
The block \(I\) is split into two consecutive blocks of sizes
\(|\{i\in I:i\le t\}|\) and \(|\{i\in I:i>t\}|\).
After multiplying by the weight of the original partition and averaging uniformly over the selected agent in \(I\), this term is exactly twice the weight of the resulting \((k+1)\)-block partition.
Conversely, every \((k+1)\)-block partition can be obtained from exactly \(k\) different \(k\)-block partitions, one for each of its \(k\) cuts that can be removed.
The same expansion used above shows that the total weight of all \((k+1)\)-block partitions is \(Z_{k+1}\).
It follows that
\begin{equation}\label{eq:gpm-refinement}
 Z_k\SC(\mathsf{GPM}_k)\le2kZ_{k+1}.
\end{equation}

It remains to compare \(Z_{k+1}\) with the optimum.
Fix an optimal solution with \(k\) facilities, and assign every agent to a nearest optimal facility, breaking ties so that the resulting clusters are consecutive.
Let \(c_i\) be the cost of agent \(i\) in this solution.
Now consider any ordered \((k+1)\)-agent subset \(U\).
Because there are only \(k\) optimal clusters, two consecutive agents \(p,q\) in \(U\) with $p<q$ belong to the same cluster.
The triangle inequality through their common facility gives
\(
 x_q-x_p\le c_p+c_q.
\)
In the product defining \(w(U)\), separate the factor \(x_q-x_p\).
If \(p\) is an interior member of \(U\), deleting it replaces the two adjacent distances meeting at \(p\) by their sum; if \(p\) is the first member of \(U\), deleting it simply removes the factor \(x_q-x_p\).
In either case, \(w(U\setminus\{p\})\) is at least the product of all factors of \(w(U)\) other than \(x_q-x_p\).
The same argument applies when \(q\) is deleted.
Therefore, assuming \(x_q-x_p>0\) (otherwise the inequality is trivial), we have
\[
 w(U)\le \frac{w(U)}{x_q-x_p}(c_p+c_q)\le c_p w(U\setminus\{p\})+c_qw(U\setminus\{q\}).
\]

Choose one such consecutive pair \(p,q\) for every \(U\), and sum the preceding inequality over all \((k+1)\)-agent subsets:
\[
\begin{aligned}
 Z_{k+1}
 &=\sum_{\substack{U:|U|=k+1}}w(U)\le
 \sum_{\substack{U:|U|=k+1}}
 \left[c_p w(U\setminus\{p\})+c_qw(U\setminus\{q\})\right].
\end{aligned}
\]
Fix an agent \(i\) and consider all summands containing the factor \(c_i\).
For every set \(U\) whose chosen pair \(p,q\) contains \(i\), the corresponding summand is
\(
 c_iw(U\setminus\{i\}),
\)
which is the first term in the preceding sum when \(i=p\), and the second term when \(i=q\).
Therefore, after factoring out \(c_i\), its coefficient is the sum of \(w(U\setminus\{i\})\) over all such sets \(U\).
Different sets \(U\) give different sets \(U\setminus\{i\}\), so these sets form a subcollection of all \(k\)-agent subsets \(S\).
Since all weights are nonnegative, the coefficient of \(c_i\) is at most
\(
 \sum_{\substack{S:|S|=k}}w(S)=Z_k.
\)

Summing this bound over all agents gives
\begin{equation}\label{eq:gpm-opt-partition}
 Z_{k+1}\le\sum_i c_iZ_k
 =\left(\sum_i c_i\right)Z_k
 =\OPT_k Z_k.
\end{equation}
Combining \eqref{eq:gpm-refinement} and \eqref{eq:gpm-opt-partition}, and dividing by \(Z_k>0\), gives
\[
 \SC(\mathsf{GPM}_k)
 \le 2k\frac{Z_{k+1}}{Z_k}
 \le 2k\OPT_k.
\]

Finally, for tightness, place \(M\) agents at each of \(0,1,\ldots,k-1\) and one agent at \(k\).
The optimum is \(1\), while 
 \(\SC(\mathsf{GPM}_k)=\frac{2kM}{M+2k-1}\to 2k
\) when $M\to \infty$.
\end{proof}

\subsection{Strategyproofness}
We now prove the incentive part of Theorem~\ref{thm:gpm-three}. When $k\ge 4$, it is not hard to find examples where agents have incentive to misreport (see, e.g., \cite{jia2026product}). When $k=2$, Gap-Product reduces to the strategyproof mechanism Pairwise-Distance. Thus, we focus only on $k=3$.

If there are at most three distinct occupied positions, the mechanism covers every occupied position, so no agent can benefit from misreporting.
We may therefore assume that the truthful profile has at least four distinct occupied positions.
A unilateral deviation changes the number of distinct occupied positions by at most one.
Consequently, if the deviating profile invokes the fallback rule, it has exactly three distinct occupied positions, and every triple of positive Gap-Product weight covers all three positions; thus the weighted formulas below agree with the fallback outcome.
As in Section~\ref{sec:pd-strategyproofness}, fix an agent at her true location \(x\), let \(z\) be an alternative report, and let \(Y\) be the labelled set of all other reports.
For every nonempty  \(S\), write
\(
 d(x,S)=\min_{y\in S}|x-y|
\)
for the minimum distance from \(x\) to any member of \(S\).
Define
\[
 B=\sum_{\substack{T\subseteq Y:|T|=3}}w(T),
 \qquad
 A=\sum_{\substack{T\subseteq Y:|T|=3}}w(T)d(x,T),
\]
and, for \(u\in\{x,z\}\),
\[
 D_u=\sum_{\substack{P\subseteq Y:|P|=2}}w(P\cup\{u\}),
 \qquad
 E=\sum_{\substack{P\subseteq Y:|P|=2}}
 w(P\cup\{z\})d(x,P\cup\{z\}).
\]
A selected triple containing the truthful report gives the agent cost zero. Therefore, her expected costs under the truthful and deviating reports are $\frac{A}{B+D_x}$ and $\frac{A+E}{B+D_z}$,
respectively. To exclude beneficial deviations, it is enough to prove
\begin{equation}\label{eq:gpm-strong}
 BE\ge A(D_z-D_x),
\end{equation}
because
\(
 E(B+D_x)-A(D_z-D_x)
 =\bigl[BE-A(D_z-D_x)\bigr]+ED_x\ge0.
\)

For a triple \(T\subseteq Y\) and a pair \(P\subseteq Y\), let
\begin{align}
 \Gamma(T,P)=w(T)\Big[&w(P\cup\{z\})
 \bigl(d(x,P\cup\{z\})-d(x,T)\bigr)+w(P\cup\{x\})d(x,T)\Big].
 \label{eq:gpm-gamma}
\end{align}
Expanding both products in \eqref{eq:gpm-strong} gives
\begin{equation}\label{eq:gpm-expand}
 BE-A(D_z-D_x)
 =\sum_{\substack{T,P\subseteq Y\\|T|=3,\ |P|=2}}\Gamma(T,P).
\end{equation}
In order to prove that it is non-negative,
we group these terms according to whether \(T\) and \(P\) share two, one, or no agents.

\paragraph{Overlapping $T$ and $P$.}
Let $a$ be a report shared by \(T\) and \(P\).
We call \(a\) the \emph{anchor}.
For any two reports \(u,v\), write
\[
 \Delta_a(u,v)=w(\{a,u,v\}).
\]
For three further reports \(u,v,s\), define
\begin{align}
 \Phi_{a,z}(u,v,s)
 =\sum_{\text{cyclic }(u,v,s)}\Delta_a(v,s)\Big[&
 \Delta_a(z,u)d(x,\{a,z,u\})-\bigl(\Delta_a(z,u)-\Delta_a(x,u)\bigr)
 d(x,\{a,v,s\})\Big].
 \label{eq:gpm-phi}
\end{align}
Here the cyclic sum means that the displayed expression is added for the three cyclic orderings $(u,v,s)$, $(v,s,u)$, and $(s,u,v)$.

\begin{lemma}\label{lem:gpm-anchor}
For all reports on the line, \(
 \Phi_{a,z}(u,v,s)\ge0.
\)
\end{lemma}

\begin{proof}
Translate the anchor \(a\) to \(0\). If \(x=0\), then the true location coincides with the anchor. Every set appearing in \eqref{eq:gpm-phi} contains this anchor, so
\(
 d(x,\{a,z,u\})=d(x,\{a,v,s\})=0.
\)
Thus every term in the cyclic sum is zero, and hence \(\Phi_{a,z}(u,v,s)=0\). Assume without loss of generality that \(x>0\). Put
\[
 \Delta(u,v)=\Delta_0(u,v),
 \qquad
 \ell(u)=\min\{x,|x-u|\}.
\]
Then \(d(x,\{0,u,v\})=\min\{\ell(u),\ell(v)\}\).
We use the following two elementary inequalities, proved in Appendix~\ref{app:gpm-aux}:
\begin{align}
 \Delta(p,u)\Delta(v,s)
 &\le \Delta(p,v)\Delta(u,s)+\Delta(p,s)\Delta(u,v),
 \label{eq:gpm-exchange}\\
 2\Delta(q,u)\Delta(q,v)
 &\ge q\rho\,\Delta(u,v),
 \label{eq:gpm-centre}
\end{align}
where the second inequality holds when \(q>0\), \(0\le\rho\le q\), and \(u,v\notin(q-\rho,q+\rho)\).

Relabel the three reports as \(u_1,u_2,u_3\) so that
\[
 \ell(u_1)=\alpha\le\ell(u_2)=\beta\le\ell(u_3),
 \qquad
 W_{pq}=\Delta(u_p,u_q).
\]
Expanding the cyclic sum gives
\begin{align}
\Phi_{0,z}
={}&W_{23}\Big[\beta\Delta(x,u_1)
 +\Delta(z,u_1)(\min\{\ell(z),\alpha\}-\beta)\Big]\notag\\
&+W_{13}\Big[\alpha\Delta(x,u_2)
 +\Delta(z,u_2)(\min\{\ell(z),\beta\}-\alpha)\Big]\notag\\
&+W_{12}\Big[\alpha\Delta(x,u_3)
 +\Delta(z,u_3)(\min\{\ell(z),\ell(u_3)\}-\alpha)\Big].
\label{eq:gpm-phi-ordered}
\end{align}

Suppose first that \(\ell(z)\ge\beta\). Rearranging \eqref{eq:gpm-phi-ordered}, we obtain
\begin{align*}
\Phi_{0,z}
={}&\beta W_{23}\Delta(x,u_1)+\alpha W_{13}\Delta(x,u_2)
 +\alpha W_{12}\Delta(x,u_3)\\
&+(\beta-\alpha)\bigl[W_{13}\Delta(z,u_2)
 +W_{12}\Delta(z,u_3)-W_{23}\Delta(z,u_1)\bigr]\\
&+W_{12}\Delta(z,u_3)
 \bigl(\min\{\ell(z),\ell(u_3)\}-\beta\bigr).
\end{align*}
Every term outside the square brackets is nonnegative, and the square bracket is nonnegative by \eqref{eq:gpm-exchange}.

It remains to consider \(\ell(z)<\beta\). Then \(\ell(z)=|x-z|\), so write \(z=x+\varepsilon r\), where \(\varepsilon\in\{-1,1\}\) and \(0\le r<\beta\). Let \(F(r)\) denote \eqref{eq:gpm-phi-ordered} for this report, and put
\[
 g_i(r)=\Delta(x+\varepsilon r,u_i).
\]
We have \(F(0)=0\), while \(F(\beta)\ge0\) by the case already proved. If \(\beta=x\) and \(\varepsilon=-1\), the endpoint report is the anchor and its value is also immediate.

For \(0\le r\le\alpha\),
\begin{align}
F(r)
={}&W_{23}[\beta g_1(0)+(r-\beta)g_1(r)]+W_{13}[\alpha g_2(0)+(r-\alpha)g_2(r)]
 +W_{12}[\alpha g_3(0)+(r-\alpha)g_3(r)].
\label{eq:gpm-F-first}
\end{align}
Since \(\ell(u_i)\ge\alpha\), no \(g_i\) has an interior breakpoint on \([-\alpha,0]\). Each \(g_i\) is therefore linear or concave quadratic there, and concavity gives
\begin{equation}\label{eq:gpm-F-zero}
F'(0+)\ge
 W_{23}[g_1(-\alpha)-(\beta-\alpha)g_1'(0)]
 +W_{13}g_2(-\alpha)+W_{12}g_3(-\alpha).
\end{equation}
This is immediate if \(\alpha=\beta\) or \(g_1'(0)\le0\). Otherwise the piecewise formula gives
\[
 u_1=x-\varepsilon\alpha,
 \quad g_1(-\alpha)=0,
 \quad g_2(-\alpha)=W_{12},
 \quad g_3(-\alpha)=W_{13},
 \quad g_1'(0)\le u_1.
\]
Here \(u_1>0\), \(\beta-\alpha\le u_1\), and \(u_2,u_3\) lie outside the interval of radius \(\beta-\alpha\) around \(u_1\). Applying \eqref{eq:gpm-centre} gives
\(
 2W_{12}W_{13}\ge(\beta-\alpha)u_1W_{23},
\)
so \eqref{eq:gpm-F-zero} yields \(F'(0+)\ge0\).

If \(\alpha=\beta\), differentiating \eqref{eq:gpm-F-first} at \(\alpha\) from the left gives \(F'(\alpha-)\ge0\), because every coefficient containing \(\beta-\alpha\) vanishes. Since \(F'\) is concave on \((0,\alpha)\), this proves \(F\ge0\) there. Now assume \(\alpha<\beta\). For \(\alpha\le r\le\beta\),
\begin{align}
F(r)
={}&W_{23}[\beta g_1(0)-(\beta-\alpha)g_1(r)]+W_{13}[\alpha g_2(0)+(r-\alpha)g_2(r)]
 +W_{12}[\alpha g_3(0)+(r-\alpha)g_3(r)].
\label{eq:gpm-F-second}
\end{align}
Let \(q=x+\varepsilon\alpha\). Since \(u_1,q\in\{x-\alpha,x+\alpha\}\) and \(\alpha<\beta\le x\), both centres are positive and the radius \(\beta-\alpha\) is admissible. The intervals of this radius around \(u_1\) and \(q\) lie inside \((x-\beta,x+\beta)\), which contains neither \(u_2\) nor \(u_3\). Applying \eqref{eq:gpm-centre} at both centres and then AM--GM gives
\begin{equation}\label{eq:gpm-alpha}
 W_{13}g_2(\alpha)+W_{12}g_3(\alpha)
 \ge(\beta-\alpha)\sqrt{u_1q}\,W_{23}.
\end{equation}
Moreover, every positive one-sided derivative \(g_1'(\alpha\pm)\) is at most \(\sqrt{u_1q}\). Explicitly,
\begin{align*}
 F'(\alpha-)
 &=W_{23}g_1(\alpha)+W_{13}g_2(\alpha)+W_{12}g_3(\alpha)
 -(\beta-\alpha)W_{23}g_1'(\alpha-),\\
 F'(\alpha+)
 &=W_{13}g_2(\alpha)+W_{12}g_3(\alpha)
 -(\beta-\alpha)W_{23}g_1'(\alpha+).
\end{align*}
Thus \eqref{eq:gpm-alpha} implies \(F'(\alpha-),F'(\alpha+)\ge0\).

Differentiating \eqref{eq:gpm-F-first} three times gives \(F'''\le0\) on \((0,\alpha)\). Hence \(F'\) is concave and nonnegative at both endpoints, so \(F(\alpha)\ge0\). Likewise, \(F'''\le0\) on \((\alpha,\beta)\). Since \(F'(\alpha+)\ge0\), \(F\) can have no strict interior minimum on this interval; its endpoint values are nonnegative. Therefore \(F(r)\ge0\) for the actual deviation. Translation and reflection complete the proof, and repeated coordinates follow by continuity.
\end{proof}

If \(|T\cap P|=1\), fix their common agent \(a\). For every fixed set of four agents $\{a,u,v,s\}$, there are three possible pairs of $T$ and $P$.
It is easy to verify that the three corresponding terms in \eqref{eq:gpm-expand} sum to 
\[
\Gamma(\{a,v,s\},\{a,u\})+\Gamma(\{a,u,s\},\{a,v\})+\Gamma(\{a,u,v\},\{a,s\})
=\Phi_{a,z}(u,v,s),\] 
which is nonnegative as desired by Lemma~\ref{lem:gpm-anchor}. 

If \(P=\{a,b\}\subseteq T=\{a,b,u\}\), then
\(
 \Phi_{a,z}(b,b,u)=2\Gamma(T,P).
\)
Thus every term $\Gamma(T,P)$ with \(|T\cap P|=2\) is also nonnegative.

\paragraph{$T$ and $P$ Disjoint.}
It remains to consider \(T\cap P=\varnothing\). For \(p,q\ge1\) and a labelled set \(U\) of \(p+q-1\) reports, define
\begin{align}
 \Psi_{p,q}(U)
 =\sum_{\substack{P\subseteq U\\|P|=p-1}}
 w(U\setminus P)\Big[&w(P\cup\{z\})
 \bigl(d(x,P\cup\{z\})-d(x,U\setminus P)\bigr)+w(P\cup\{x\})d(x,U\setminus P)\Big].
 \label{eq:gpm-psi}
\end{align}
Fix a five-agent set \(U\) so that it is the union of disjoint sets \(T\) and \(P\), and choosing the two-agent set \(P\subseteq U\) uniquely determines the triple \(T=U\setminus P\).
Therefore, the sum of all terms in \eqref{eq:gpm-expand} corresponding to $U$ is
\begin{align*}
\sum_{\substack{P\subseteq U\\|P|=2}}\Gamma(U\setminus P,P)
&=
\sum_{|P|=2}
w(U\setminus P)w(P\cup\{z\})
\bigl(d(x,P\cup\{z\})-d(x,U\setminus P)\bigr)\\
&\quad+
\sum_{|P|=2}w(U\setminus P)w(P\cup\{x\})d(x,U\setminus P)\\
&=\Psi_{3,3}(U).
\end{align*}
Hence, it suffices to prove that \(\Psi_{3,3}(U)\ge0\) for every five-agent set \(U\).

The argument uses the following five expressions with different sizes of \(U\):
\begin{align*}
 H_0(U)&=\Psi_{3,3}(U), && |U|=5,\\
 H_1(U)&=\Psi_{2,3}(U)+\Psi_{3,2}(U), && |U|=4,\\
 H_2(U)&=\Psi_{1,3}(U)+2\Psi_{2,2}(U)+\Psi_{3,1}(U), && |U|=3,\\
 {H_3(U)}&={3\Psi_{1,2}(U)+3\Psi_{2,1}(U)
 +w(U\cup\{x\})-w(U\cup\{z\})}, && {|U|=2},\\
 {H_4(U)}&={6\Psi_{1,1}(U)
 +4\bigl(w(U\cup\{x\})-w(U\cup\{z\})\bigr)}, && {|U|=1}.
\end{align*}

\begin{lemma}\label{lem:gpm-five}
For every labelled five-agent set \(U\),
\[
 \Psi_{3,3}(U)=H_0(U)\ge0.
\]
\end{lemma}

\begin{proof}
We use three elementary facts. Their computations are given in Appendices~\ref{app:gpm-complementary} and \ref{app:gpm-boundary}.

First, for a labelled set \(U\) of \(a+b\) reports, define
\[
 \Lambda_{a,b}(U)=
 \sum_{\substack{P\subseteq U\\|P|=a}}
 w(P)w(U\setminus P)
 \bigl(d(x,P)-d(x,U\setminus P)\bigr).
\]
Appendix~\ref{app:gpm-complementary} proves
\begin{equation}\label{eq:gpm-lambda}
 \Lambda_{1,2}(U)\ge0,
 \qquad \Lambda_{1,3}(U)\ge0,
 \qquad \Lambda_{2,3}(U)\ge0.
\end{equation}

Second, translate \(x\) to \(0\). If an ordinary report \(y\) is the unique outermost report and moves outward along its ray, direct differentiation gives
\begin{equation}\label{eq:gpm-ordinary-derivative}
 \frac{d}{ds}H_j(U)=H_{j+1}(U\setminus\{y\}),
 \qquad j=0,1,2,3,
\end{equation}
where \(s=|y|\). If the deviating report \(z\) is uniquely outermost and moves outward, then
\begin{align}
 H_0'(U)&=\Lambda_{2,3}(U), && |U|=5,\notag\\
 H_1'(U)&=\Lambda_{1,3}(U), && |U|=4,\notag\\
 H_2'(U)&=w(U)+\Lambda_{1,2}(U), && |U|=3,\notag\\
 H_3'(U)&=2w(U), && |U|=2.
 \label{eq:gpm-z-derivative}
\end{align}
Every derivative in \eqref{eq:gpm-z-derivative} is nonnegative by \eqref{eq:gpm-lambda}.

Third, Appendix~\ref{app:gpm-boundary} proves the following boundary fact: if \(H_{j+1}\) is nonnegative for all configurations and the largest distance from \(x\) is attained by at least two reports, then \(H_j\ge0\). For \(j=4\), no induction hypothesis is needed. The proof treats co-located ordinary reports by Lemma~\ref{lem:gpm-anchor} and handles opposite outer reports by exact derivative calculations; it does not move co-located labels independently.

We now prove \(H_4,H_3,H_2,H_1,H_0\ge0\) in this order. For \(U=\{y\}\),
\[
 H_4(U)=6|x-z|+4(|x-y|-|z-y|)
 \ge2|x-z|\ge0.
\]
Assume \(H_{j+1}\ge0\) everywhere and consider \(H_j\), where \(j=3,2,1,0\). If the outer radius is tied, apply the boundary fact. Otherwise move the unique outermost report inward until the first tie or until it reaches \(x\). Reversing this motion has nonnegative derivative by \eqref{eq:gpm-ordinary-derivative} or \eqref{eq:gpm-z-derivative}, so moving inward cannot increase \(H_j\). The endpoint is nonnegative by the boundary fact, or it is the completely collapsed zero configuration. Hence the original value is nonnegative. Descending to \(H_0\) proves the lemma.
\end{proof}

\paragraph{Completing the proof.}
If the truthful profile has at most three distinct occupied positions, the mechanism covers all positions and every agent has zero cost. 
Therefore, assume that there are at least four distinct positions. By the reduction above, it is enough to prove that the expression $\sum_{\substack{T,P\subseteq Y:|T|=3,\ |P|=2}}\Gamma(T,P)$ in \eqref{eq:gpm-expand} is nonnegative.
If \(T\cap P=\varnothing\), their union is a five-agent set \(U\).
For a fixed \(U\), the sum of all terms in \eqref{eq:gpm-expand} corresponding to \(U\) is
$
\sum_{\substack{P\subseteq U:|P|=2}}\Gamma(U\setminus P,P)
=\Psi_{3,3}(U)$, which is nonnegative by Lemma~\ref{lem:gpm-five}.
If \(T\cap P\ne\varnothing\), the groups of terms are also nonnegative by Lemma~\ref{lem:gpm-anchor}. Thus every group in \eqref{eq:gpm-expand} is nonnegative, so truthful reporting minimizes the agent's expected cost.

\section{Conclusion}\label{sec:conclusion}

This paper develops three distance-weighted randomization rules for multi-facility location.
For two facilities, Pairwise-Distance is strategyproof on Ptolemaic spaces and has a tight approximation ratio of \(4\), while Hybrid-Distance, the optimal fixed-probability mixture of Pairwise-Distance and Proportional, is strategyproof on the same spaces and improves the tight ratio to \(3.5186\).
For \(n\) agents and \(k=n-1\) facilities, Inverse-Square is strategyproof on every metric space and has a tight approximation ratio of \(\Theta(\sqrt{n})\), improving the previous best-known ratio of \(\frac{n}{2}\).
On the line, Gap-Product is strategyproof for \(k=3\) with a tight approximation ratio of \(6\), but is not strategyproof for any \(k\ge4\).
Together, these results show how distance-weighted randomization can be matched to Ptolemaic geometry, randomized omission when \(k=n-1\), and the ordered gap structure of the line.

Several central questions remain open.
For two facilities, it is unknown whether a randomized strategyproof mechanism outside the fixed Proportional--Pairwise-Distance family can improve the ratio achieved here.
For \(k\ge4\) on the line, the next challenge is to find a different randomized strategyproof rule with guarantees that scale well beyond three facilities.
In arbitrary metric spaces, extending inverse-nearest-neighbor randomization from one omitted report to several omitted reports may provide a route from the endpoint \(k=n-1\) toward general \(n\) and \(k\).
Bridging these regimes would move toward a unified account of the approximation cost of strategyproofness in randomized multi-facility location.

\section*{AI Use Statement}
We used OpenAI GPT-based tools, including ChatGPT and Codex, to assist in developing and checking mechanisms, intermediate lemmas, and proofs.
For the two-facility results, the authors supplied the central idea of combining two mechanisms with complementary behavior, while these tools played the primary role in developing the detailed strategyproofness proofs and approximation-ratio analyses.
For the results with more than two facilities, these tools played the primary role in developing the mechanisms, lemmas, and proofs, including both the incentive and approximation analyses.
The authors independently checked all AI-assisted claims, proofs, and calculations and take full responsibility for their correctness and for the paper's content.

\bibliographystyle{plain}
\bibliography{mybibfile}

\appendix

\renewcommand{\thesection}{Appendix \Alph{section}}
\renewcommand{\thesubsection}{\Alph{section}.\arabic{subsection}}
\renewcommand{\theequation}{\Alph{section}.\arabic{equation}}
\section{Auxiliary Proofs for Gap-Product}\label{app:gpm}

\subsection{Elementary inequalities for Gap-Product}\label{app:gpm-aux}

This appendix proves the two gap inequalities used in Lemma~\ref{lem:gpm-anchor}. After translating the anchor to \(0\), write
\(
 \Delta(u,v)=w(\{0,u,v\}).
\)
For fixed \(v>0\), direct expansion gives
\[
 \Delta(u,v)=
 \begin{cases}
  (-u)v, & u\le0,\\
  u(v-u), & 0\le u\le v,\\
  v(u-v), & u\ge v.
 \end{cases}
\]
The case \(v<0\) follows by reflection. In particular, \(u\mapsto\Delta(u,v)\) is linear or concave quadratic between consecutive breakpoints.

\paragraph{Four-point exchange.}
We first prove that, for all \(p,u,v,s\in\mathbb R\),
\[
 \Delta(p,u)\Delta(v,s)
 \le \Delta(p,v)\Delta(u,s)+\Delta(p,s)\Delta(u,v).
\]
If one point is \(0\), the inequality is immediate. Otherwise define
\[
 \eta(u,v)=\frac{\Delta(u,v)}{|u||v|}.
\]
This quantity equals \(1\) when \(u\) and \(v\) lie on opposite rays. When they lie on the same ray, it equals
\(
 1-\frac{\min\{|u|,|v|\}}{\max\{|u|,|v|\}}.
\)
It is a metric. Indeed, for \(0<u\le v\le s\), we have
\(
 1-\eta(u,s)=(1-\eta(u,v))(1-\eta(v,s)),
\)
which implies the triangle inequality, and cases involving both rays are immediate. Dividing the desired exchange inequality by \(|puvs|\) therefore proves every mixed-ray case. If all four points lie on one ray, order their radii as \(u_1\le u_2\le u_3\le u_4\). The only possibly largest matching product is
\(
 \eta(u_1,u_3)\eta(u_2,u_4).
\)
If the other two matching products are denoted by \(A\) and \(C\), their excess over this product is
\[
 A+C-\eta(u_1,u_3)\eta(u_2,u_4)
 =\frac{(u_2-u_1)(u_3-u_2)(u_4-u_3)}{u_2u_3u_4}\ge0.
\]

\paragraph{One-centre inequality.}
Let \(q>0\), \(0\le\rho\le q\), and \(u,v\notin(q-\rho,q+\rho)\). We prove
\[
 2\Delta(q,u)\Delta(q,v)\ge q\rho\Delta(u,v).
\]
The cases \(\rho=0\) or \(uv=0\) are immediate. Put
\[
 \lambda=\frac{\rho}{q},
 \qquad U=\eta(q,u),
 \qquad V=\eta(q,v),
\]
and assume \(u\le v\). After division by \(q^2|u||v|\), it remains to prove \(2UV\ge\lambda\eta(u,v)\).

If \(u,v\le q-\rho\), then \(U,V\ge\lambda\) and \(\eta(u,v)\le U+V\). If \(u\le q-\rho\le q+\rho\le v\), then
\[
 U\ge\lambda,
 \qquad V\ge\frac{\lambda}{1+\lambda},
 \qquad \eta(u,v)=U+V-UV.
\]
For fixed \(V\), the desired difference increases with \(U\) and is nonnegative at \(U=\lambda\). Finally, if \(q+\rho\le u\le v\), then \(U\ge\lambda/(1+\lambda)\ge\lambda/2\) and \(\eta(u,v)\le V\). These cases prove the inequality.

\subsection{Complementary-sample inequalities}\label{app:gpm-complementary}

We prove the three inequalities in \eqref{eq:gpm-lambda}. For \(\Lambda_{1,2}\), order the three reports by their distance from \(x\), say \(d_1\le d_2\le d_3\), and denote their coordinates by \(y_1,y_2,y_3\). Expanding the three choices of the singleton gives
\begin{align*}
\Lambda_{1,2}
={}&-(d_2-d_1)|y_2-y_3|
 +(d_2-d_1)|y_1-y_3|+(d_3-d_1)|y_1-y_2|\\
\ge{}&(d_3-d_2)|y_1-y_2|\ge0,
\end{align*}
where the first inequality is the triangle inequality.

For \(\Lambda_{1,3}\), order four reports by distance from \(x\), and let \(W_i\) be the weight of the triple obtained by deleting the \(i\)-th report in this distance order. Then
\[
 \Lambda_{1,3}\ge(d_2-d_1)(W_2+W_3+W_4-W_1).
\]
To see that the last factor is nonnegative, order the same four reports by coordinate and call their consecutive gaps \(a,b,c\). The four deletion weights are $bc$, $(a+b)c$, $a(b+c)$ and $ab$.
Each of these four numbers is at most the sum of the other three: after cancellation, the only nontrivial differences are \(2ab\) and \(2bc\). Hence \(\Lambda_{1,3}\ge0\).

For \(\Lambda_{2,3}\), use the layer-cake expression
\[
 d(x,S)=\int_0^\infty \mathbf 1\{S\cap N_s=\varnothing\}\,ds,
 \qquad
 N_s=\{u\in U:|u-x|\le s\}.
\]
For a fixed \(N=N_s\), the integrand of \(\Lambda_{2,3}\) is
\begin{align*}
 G(N)=\sum_{\substack{P\subseteq U\\|P|=2}}w(P)w(U\setminus P)
 \Big(&\mathbf 1\{P\cap N=\varnothing\}-\mathbf 1\{(U\setminus P)\cap N=\varnothing\}\Big).
\end{align*}
It is zero when \(N=\varnothing\), and it is nonnegative when \(|N|\ge3\): the second indicator is then always zero, while the first one is nonnegative. Since \(N_s\) is a consecutive block in coordinate order, only five cases remain up to reflection. Let \(y_1\le\cdots\le y_5\) have consecutive gaps \(a,b,c,d\). Direct expansion gives
\[
\begin{array}{c|l}
N&G(N)\\
\hline
\{y_1\}&2abc+2abd+4ac^2+2acd\\
\{y_2\}&2abc+2abd\\
\{y_3\}&2abc+4b^2c+4bc^2+2bcd\\
\{y_1,y_2\}&2abc+2abd+2ac^2+acd\\
\{y_2,y_3\}&2abc+abd+2b^2c+2bc^2+bcd
\end{array}
\]
Every coefficient is nonnegative. Integrating over \(s\) proves \(\Lambda_{2,3}\ge0\).

\subsection{Boundary configurations in the radial argument}\label{app:gpm-boundary}

We prove the boundary fact used in Lemma~\ref{lem:gpm-five}.

\paragraph{Far reports and collisions.}
If a bounded set \(U\) is augmented by an ordinary report \(L\to+\infty\), then
\begin{equation}\label{eq:gpm-far}
 H_j(U\cup\{L\})=L H_{j+1}(U)+O(1),
 \qquad j=0,1,2,3.
\end{equation}
Indeed, \(w(S\cup\{L\})=(L-\max S)w(S)\), while \(L\) is never nearest to \(x\) in a selected set containing a bounded report.

If two ordinary labels coincide at \(q\), then
\(
 H_0(\{q,q,u,v,s\})=2\Phi_{q,z}(u,v,s)\ge0.
\)
Applying \eqref{eq:gpm-far} successively gives the corresponding statement for \(H_1,H_2,H_3\).

If \(z\) coincides with an ordinary report, split each sum \(\Psi_{p,q}\) according to the side containing that ordinary copy. On the \(P\)-side the deviating weight vanishes. On the complementary side, pairing \(P\) with its complement cancels the deviation parts of \(\Psi_{p,q}+\Psi_{q,p}\), while every truthful part remains nonnegative.
 {More explicitly, the two additional terms are
\[
 w(U\cup\{x\})-w(U\cup\{z\}) \quad\text{in }H_3(U),
 \qquad
 4\bigl(w(U\cup\{x\})-w(U\cup\{z\})\bigr) \quad\text{in }H_4(U).
\]
Because \(z\) coincides with an ordinary report in \(U\), the set \(U\cup\{z\}\) contains two reports at the same position. Therefore \(w(U\cup\{z\})=0\), whereas \(w(U\cup\{x\})\ge0\), so both displayed terms are nonnegative.}
Thus every \(H_j\) is nonnegative at such a collision.

\paragraph{Opposite outer reports involving the deviation.}
Translate \(x\) to \(0\), let \(z=R\), and put one ordinary report at \(-R\), with all other ordinary reports strictly inside. If \(V\) denotes the inner ordinary reports, differentiating as both outer reports move outward gives
\begin{align*}
 \frac{d}{dR}H_0(\{-R\}\cup V)
 &=H_1(V)+\Lambda_{2,3}(\{-R\}\cup V), && |V|=4,\\
 \frac{d}{dR}H_1(\{-R\}\cup V)
 &=H_2(V)+\Lambda_{1,3}(\{-R\}\cup V)+2R w(V), && |V|=3,\\
 \frac{d}{dR}H_2(\{-R\}\cup V)
 &=H_3(V)+w(\{-R\}\cup V)+\Lambda_{1,2}(\{-R\}\cup V)
   +4R w(V), && |V|=2,\\
 \frac{d}{dR}H_3(\{-R\}\cup V)
 &=H_4(V)+2w(\{-R\}\cup V)+6R, && |V|=1.
\end{align*}
Each identity follows by differentiating the last gap. A selected set containing both outer reports contributes the additional displayed term involving \(R\). Every right-hand side is nonnegative by the radial induction and \eqref{eq:gpm-lambda}. The remaining boundary case is direct:
\(
 H_4(\{-R\})=2R\ge0.
\)
Moving the two outer reports inward reaches a collision already handled above.

\paragraph{Two opposite ordinary reports.}
Now the ordinary reports include \(-R\) and \(R\), while \(|z|<R\). If \(V\) denotes the three inner ordinary reports, then
\[
 \frac{d}{dR}H_0(\{-R,R\}\cup V)
 =H_1(\{R\}\cup V)+H_1(\{-R\}\cup V)\ge0.
\]
At the other end of the induction,
\[
 \frac{d}{dR}H_3(\{-R,R\})=12|z|\ge0.
\]

For \(H_2\), let \(y\) be the remaining ordinary report and reflect so that \(0\le y<R\). Direct differentiation gives
\[
\frac{d}{dR}H_2(\{-R,R,y\})=
\begin{cases}
 z^2-8zy+6y^2, & -R\le z\le-y,\\
 -z(-7z+8y), & -y\le z\le0,\\
 z(24R-5z+4y), & 0\le z\le y,\\
 24Ry+z^2+4zy-6y^2, & y\le z\le R.
\end{cases}
\]
In the first interval, writing \(z=-y-u\) turns the expression into \(u^2+10uy+15y^2\). The middle two expressions are visibly nonnegative, and the last is at least \(23y^2\) because \(R\ge z\ge y\).

It remains to verify the derivative of \(H_1\). Interchange the two inner ordinary reports and reflect if necessary so that \(y_1\le y_2\) and \(y_1+y_2\ge0\). Let \(\chi=1\) when \(y_1<0\), let \(\chi=0\) otherwise, and put \(\sigma=1-2\chi\). In the following table all displayed variables are nonnegative. The substitutions cover the six consecutive intervals containing \(z\), and each final column is the exact expansion of \(dH_1/dR\).

\scriptsize
\renewcommand{\arraystretch}{1.18}
\begin{longtable}{c|p{0.30\textwidth}|p{0.58\textwidth}}
\toprule
&Substitution&Exact nonnegative expansion of \(dH_1/dR\)\\
\midrule
\endhead
1&\(y_1=\sigma a, y_2=a+b, z=-y_2-u, R=y_2+u+v\)
&\(10a^3+27a^2b+12a^2u+12a^2v+18ab^2+18abu+14abv+4b^3+7b^2u+8b^2v+2bu^2+2buv\)
\newline{}\(\quad+\chi(12a^3+28a^2b+20a^2u+16a^2v+16ab^2+20abu+24abv+4au^2+4auv)\)\\[1mm]
2&\(y_1=\sigma a, z=-a-v, y_2=a+u+v, R=y_2+h\)
&\(10a^3+12a^2h+19a^2u+27a^2v+2ahu+14ahv+16auv+18av^2+2huv+8hv^2+5uv^2+4v^3\)
\newline{}\(\quad+\chi(12a^3+16a^2h+24a^2u+28a^2v+20ahu+24ahv+20au^2+36auv+16av^2)\)\\[1mm]
3&\(y_1=\sigma(u+v), z=-v, y_2=u+v+b, R=y_2+h\)
&\(v[2bh+6bu+19bv+12hv+2u^2+14uv+10v^2\)
\newline{}\(\qquad+\chi(20b^2+20bh+36bu+24bv+24hu+16hv+18u^2+32uv+12v^2)]\)\\[1mm]
4&\(y_1=\sigma(u+v), z=u, y_2=u+v+b, R=y_2+h\)
&\(u[2bh+31bu+22bv+12hu+24hv+10u^2+30uv+18v^2\)
\newline{}\(\qquad+\chi(16hu+12u^2+16uv+2v^2)+(1-\chi)(20b^2+20bh+20bv)]\)\\[1mm]
5&\(y_1=\sigma a, z=a+u, y_2=a+u+v, R=y_2+h\)
&\(10a^3+12a^2h+31a^2u+31a^2v+22ahu+2ahv+20au^2+24auv+2huv+u^2v\)
\newline{}\(\quad+\chi(12a^3+16a^2h+12a^2u)+(1-\chi)(20ahv+20auv+20av^2)\)\\[1mm]
6&\(y_1=\sigma a, y_2=a+b, z=y_2+u, R=y_2+u+v\)
&\(10a^3+31a^2b+12a^2u+12a^2v+20ab^2+26abu+22abv+3b^2u+2bu^2+2buv\)
\newline{}\(\quad+\chi(12a^3+12a^2b+20a^2u+16a^2v+4abu+4au^2+4auv)\)\\
\bottomrule
\end{longtable}
\normalsize

When \(\chi=1\), the six intervals are
\[
 [-R,-y_2],\ [-y_2,y_1],\ [y_1,0],\ [0,-y_1],\ [-y_1,y_2],\ [y_2,R].
\]
When \(\chi=0\), they are
\[
 [-R,-y_2],\ [-y_2,-y_1],\ [-y_1,0],\ [0,y_1],\ [y_1,y_2],\ [y_2,R].
\]
Every coefficient in the table is nonnegative, so \(dH_1/dR\ge0\).

Thus every opposite-outer configuration is nondecreasing as the two outer reports move away from the origin. Moving them inward reaches either a collision, already proved nonnegative, or the completely collapsed configuration. This proves the boundary fact.

\end{document}